\documentclass[11pt,a4paper]{article}

\usepackage{amsmath,amsfonts,amssymb,amsthm}
\usepackage{mathbbol}
\usepackage{amsthm}
\usepackage{graphicx,color}
\usepackage{boxedminipage}
\usepackage[ruled, linesnumbered, boxed,lined]{algorithm2e}
\usepackage{algorithmicx}
\usepackage{framed}
\usepackage{thmtools}
\usepackage{thm-restate}
\usepackage{xspace}
\usepackage{vmargin}
\setmarginsrb{1in}{1in}{1in}{1in}{0pt}{0pt}{0pt}{6mm}
\usepackage{todonotes}
 \usepackage[pdftex, plainpages = false, pdfpagelabels, 
                 bookmarks=false,
                 bookmarksopen = true,
                 bookmarksnumbered = true,
                 breaklinks = true,
                 linktocpage,
                 pagebackref,
                 colorlinks = true,  
                 linkcolor = blue,
                 urlcolor  = blue,
                 citecolor = red,
                 anchorcolor = green,
                 hyperindex = true,
                 hyperfigures
                 ]{hyperref} 
 \usepackage{xifthen}
 \usepackage{tabularx}
\usepackage{arydshln}
 \usetikzlibrary{calc}

 \usepackage{mathtools}
\DeclarePairedDelimiter\ceil{\lceil}{\rceil}
\DeclarePairedDelimiter\floor{\lfloor}{\rfloor}

\newcommand{\hdist}{d_H}

\newcommand{\Oh}{\mathcal{O}}

\newcommand{\cost}{{\sf cost}}

\DeclareMathOperator{\operatorClassNP}{{\sf NP}}
\newcommand{\classNP}{\ensuremath{\operatorClassNP}}
\DeclareMathOperator{\operatorClassCoNP}{{\sf coNP}}
\newcommand{\classCoNP}{\ensuremath{\operatorClassCoNP}}
\DeclareMathOperator{\operatorClassFPT}{{\sf FPT}\xspace}
\newcommand{\classFPT}{\ensuremath{\operatorClassFPT}\xspace}
\DeclareMathOperator{\operatorClassW}{{\sf W}}
\newcommand{\classW}[1]{\ensuremath{\operatorClassW[#1]}}

\newcommand{\bfA}{\mathbf{A}}

\newcommand{\bfa}{\mathbf{a}} 
\newcommand{\bfb}{\mathbf{b}} 
\newcommand{\bfc}{\mathbf{c}}

\newcommand{\bfS}{\mathbf{S}} 
\newcommand{\bfs}{\mathbf{s}}

\newtheorem{theorem}{Theorem}
\newtheorem{lemma}{Lemma}
\newtheorem{claim}{Claim}[section]
\newtheorem{corollary}{Corollary}

\newtheorem{observation}{Observation}
\newtheorem{proposition}{Proposition}

\newcommand{\pname}{\textsc}
\newcommand{\ProblemFormat}[1]{\pname{#1}}
\newcommand{\ProblemIndex}[1]{\index{problem!\ProblemFormat{#1}}}
\newcommand{\ProblemName}[1]{\ProblemFormat{#1}\ProblemIndex{#1}{}\xspace}

 \newcommand{\probClust}{\ProblemName{Categorical Clustering}}
\newcommand{\probBClust}{\ProblemName{Balanced Clustering}}
\newcommand{\probFBClust}{\ProblemName{Factor-Balanced Clustering}}
\newcommand{\probCClust}{\ProblemName{Capacitated Clustering}}
\newcommand{\probEClust}{\ProblemName{Equal Clustering}}

\newcommand{\probPM}{\ProblemName{Minimum Weight Perfect Matching}}
\newcommand{\probMC}{\ProblemName{Max-Cut}}



\newlength{\RoundedBoxWidth}
\newsavebox{\GrayRoundedBox}
\newenvironment{GrayBox}[1]%
   {\setlength{\RoundedBoxWidth}{.93\textwidth}
    \def\boxheading{#1}
    \begin{lrbox}{\GrayRoundedBox}
       \begin{minipage}{\RoundedBoxWidth}}%
   {   \end{minipage}
    \end{lrbox}
    \begin{center}
    \begin{tikzpicture}%
       \node(Text)[draw=black!20,fill=white,rounded corners,%
             inner sep=2ex,text width=\RoundedBoxWidth]%
             {\usebox{\GrayRoundedBox}};
        \coordinate(x) at (current bounding box.north west);
        \node [draw=white,rectangle,inner sep=3pt,anchor=north west,fill=white] 
        at ($(x)+(6pt,.75em)$) {\boxheading};
    \end{tikzpicture}
    \end{center}}     

\newenvironment{defproblemx}[2][]{\noindent\ignorespaces%
                                \FrameSep=6pt%
                                \parindent=0pt%
                \vspace*{-1.5em}
                \ifthenelse{\isempty{#1}}{%
                  \begin{GrayBox}{\textsc{#2}}%
                }{%
                  \begin{GrayBox}{\textsc{#2} parameterized by~{#1}}%
                }
                \begin{tabular*}{\textwidth}{@{\hspace{.1em}} >{\itshape} p{1.8cm} p{0.8\textwidth} @{}}%
            }{
                \end{tabular*}%
                \end{GrayBox}%
                \ignorespacesafterend
            }

\newcommand{\defproblema}[3]{
  \begin{defproblemx}{#1}
    Input:  & #2 \\
    Task: & #3
  \end{defproblemx}
}%

\begin{document}

\title{Parameterized Complexity of Categorical Clustering with Size Constraints\thanks{A preliminary version of the paper is accepted for WADS 2021. The research leading to these results have  been supported by the Research Council of Norway via the project ``MULTIVAL" (grant no. 263317) and the European Research Council (ERC) via grant LOPPRE, reference 819416.}
}

\author{
Fedor V. Fomin\thanks{
Department of Informatics, University of Bergen, Norway.} \addtocounter{footnote}{-1}
\and
Petr A. Golovach\footnotemark{} \addtocounter{footnote}{-1}
\and 
Nidhi Purohit\footnotemark{} 
}

\date{}

\maketitle

\begin{abstract}
In the \probClust problem, we are given a set of vectors (matrix) $\bfA=\{\bfa_1,\ldots,\bfa_n\}$ over $\Sigma^m$, where $\Sigma$ is a finite alphabet, and integers $k$ and $B$.
The task is to partition $\bfA$ into $k$ clusters such that the median objective of the clustering in the Hamming norm is at most $B$. That is, we seek a  partition $\{I_1,\ldots,I_{k}\}$ of $\{1,\ldots,n\}$ and vectors $\bfc_1,\ldots,\bfc_{k}\in\Sigma^m$ such that 
$$\sum_{i=1}^{k}\sum_{j\in I_i}\hdist(\bfc_i,\bfa_j)\leq B,$$
where $\hdist(\bfa,\bfb)$ is the Hamming distance between vectors $\bfa$ and $\bfb$.
Fomin, Golovach, and Panolan [ICALP 2018] proved that the problem is fixed-parameter tractable (for binary case $\Sigma=\{0,1\}$) by  giving an algorithm that solves the problem in time $2^{\Oh(B\log B)}\cdot (mn)^{\Oh(1)}$. 

We extend this algorithmic result  to a popular capacitated clustering model, where in addition the sizes of the clusters should satisfy certain constraints. More precisely, 
in  \probCClust, in addition,  we are given two non-negative integers $p$ and $q$, and seek a   clustering with  $p\leq |I_i|\leq q$ for all $i\in\{1,\ldots,k\}$.  
Our main theorem is that  
 \probCClust is solvable in time $2^{\Oh(B\log B)}|\Sigma|^B\cdot (mn)^{\Oh(1)}$. 
 The theorem not only   extends the previous algorithmic results to a significantly more general model, it also  implies algorithms for several other variants of \probClust with constraints on cluster sizes. 
 \end{abstract}

\section{Introduction}\label{sec:intro}
 While many problems in machine learning concerns numerical data, there is a large class of problems about learning from categorical data.
 The term categorical data refers to the type of data whose values are discrete and belong to a specific finite set of categories.  It could be text, some numeric values,  or even unstructured data like images.
 The most popular clustering objectives for numerical data are $k$-means and $k$-median, that 
 are based on distances in 
  the $\ell_1$ and $\ell_2$-norm. For categorical data, other meters, like Hamming distance, could be much more useful.

We study the parameterized complexity of clustering problems with constraints on the sizes of the clusters.  
The need for  clustering  with constraints comes from various application. The survey of 
Banerjee  and Ghosh  \cite{banerjee2008clustering} contains a number of examples of clustering with balancing constraints in Direct Marketing \cite{DBLP:conf/icdm/YangP03}, Category Management \cite{nielsen1992category},  Clustering of Documents \cite{baeza1999modern,lynch1999web}, and  
Energy Aware Sensor Networks  \cite{ghiasi2002optimal,gupta2003load}   among others.  However, introducing constraints on the sizes of clustering usual makes clustering tasks much more computationally challenging. 

In this paper we focus on  categorical data clustering, where data features  admit a fixed number of possible values.
We work with vectors from  $\Sigma^m$, where 
 $\Sigma$ is a finite alphabet. The most commonly used similarity measure for categorical data is the Hamming distance. For two vectors $\bfa,\bfb\in\Sigma^m$ or, equivalently, for two strings of length $m$ over $\Sigma$, we use $\hdist(\bfa,\bfb)$, to denote the \emph{Hamming distance} between $\bfa$ and $\bfb$, that is, the number of indices $i\in\{1,\ldots,m\}$ where the $i$-th elements of $\bfa$ and $\bfb$ differ. 
The task of the vanilla \probClust problem is, given an $m\times n$ matrix $\bfA$ with columns $(\bfa_1,\ldots,\bfa_n)$ over a finite alphabet $\Sigma$, a positive integer $k$, and a nonnegative integer $B$, decide
whether there is  a partition $\{I_1,\ldots,I_{k}\}$ of $\{1,\ldots,n\}$ and vectors $\bfc_1,\ldots,\bfc_{k}\in\Sigma^m$ such that 
$$\sum_{i=1}^{k}\sum_{j\in I_i}\hdist(\bfc_i,\bfa_j)\leq B.$$
The sets $I_1,\ldots,I_k$ are called \emph{clusters} and the vectors $\bfc_1,\ldots,\bfc_k$ are \emph{medians} (or  \emph{centers})\footnote{Some authors call $\bfc_1,\ldots,\bfc_k$ \emph{means} in the case of Hamming distances.}.  
We consider the generalization of the problem, where the size of each cluster should be within a given interval:

\defproblema{\probCClust}%
{An $m\times n$ matrix $\bfA$ with columns $(\bfa_1,\ldots,\bfa_n)$ over a finite alphabet $\Sigma$, a positive integer $k$, a nonnegative integer $B$, and positive integers $p$ and $q$ such that $p\leq q$.}%
{Decide whether there is  a partition $\{I_1,\ldots,I_{k}\}$ of $\{1,\ldots,n\}$, where $p\leq |I_i|\leq q$, and vectors $\bfc_1,\ldots,\bfc_{k}\in\Sigma^m$ such that 
$$\sum_{i=1}^{k}\sum_{j\in I_i}\hdist(\bfc_i,\bfa_j)\leq B.$$}

Parameterized algorithms for the vanilla variant of  \probCClust (without constraints on the sizes of clusters) were given by  Fomin, Golovach and Panolan in \cite{FominGP20}. One of the main results of their paper is the theorem providing an algorithm of running time $2^{\Oh(B\log B)}\cdot (nm)^{\Oh(1)}$ for vanilla clustering over binary field. In other words, the problem is fixed-parameter tractable (FPT) parameterized by $B$. The main question that we address in this paper is whether clustering constraints impact the problem's parameterized complexity.

\paragraph*{Our results.} 
Our main result is that \probCClust 
is fixed-parameter tractable when parameterized by the budget $B$ and the alphabet size. More precisely, we show the following:

\begin{restatable}{theorem}{main}\label{thm:FPT-k}
 \probCClust can be solved in $2^{\Oh(B\log B)}|\Sigma|^B\cdot (mn)^{\Oh(1)}$  time. 
 \end{restatable}

Fomin, Golovach and Panolan~\cite[Theorem~1]{FominGP20} proved that \probClust for binary matrices is \classFPT when parameterized  by the budget $B$. Theorem~\ref{thm:FPT-k}  generalizes this result. Interestingly, for approximation algorithms, introducing clustering constraints makes the problem much more computationally challenging. However, from parameterized complexity perspective, adding constraints does not change the complexity of the problem. 

We also observe  that \probCClust is \classNP-complete even for binary matrices, $k=2$ and $p=q=\frac{n}{k}$. Theorem~\ref{thm:FPT-k} can be used to establish fixed-parameter tractability of several  other variants of constrained clustering discussed in the literature.  In  some applications, it is natural to require that the sizes of clusters should be approximately equal, see e.g. \cite{Vallejo-HuangaM17}. We consider 
 variants of \probClust, where the input contains additional parameters besides a matrix $\bfA=(a_1,\ldots,a_n)$ and integers $k$ and $B$, and the task is to find clusters $I_1,\ldots,I_k$ and medians $\bfc_1,\ldots,\bfc_{r}\in\Sigma^m$
such that $\sum_{i=1}^{k}\sum_{j\in I_i}\hdist(\bfc_i,\bfa_j)\leq B$ and the sizes of the clusters satisfy special balance properties.
\begin{itemize}
\item In \probBClust, we are additionally given a nonnegative integer $\delta$ and it should hold that $||I_i|-|I_j||\leq \delta$ for all $i,j\in\{1,\ldots,k\}$, that is, the sizes of  clusters can differ by at most $\delta$.
\item In \probFBClust, we are given  a real $\alpha\geq 1$ and it is required that $|I_i|\leq \alpha|I_j|$ for all $i,j\in\{1,\ldots,k\}$, that is, the ratio of the clusters sizes is upper bounded by $\alpha$.
\end{itemize} 

By making use of Theorem~\ref{thm:FPT-k}, we prove that  
\probBClust and \probFBClust are solvable in time $2^{\Oh(B\log B)}|\Sigma|^B\cdot (mn)^{\Oh(1)}$. 

We conclude by discussing kernelization for these problems. In particular, we show that  \probBClust admits a polynomial kernel under the combined parameterization by $k$, $B$ and $\delta$. 
We also observe that neither of considered problems has a polynomial kernel when parameterized by $B$ only, unless $\classCoNP\subseteq\classNP/{\sf poly}$, even for the binary case.

\medskip\noindent\textbf{High-level overview of the proof of Theorem~\ref{thm:FPT-k}.} The algorithm for the vanilla problem of Fomin et al. ~\cite{FominGP20}, as well as the algorithm  of Fomin, Golovach and Simonov for clustering in $\ell_p$-norn~\cite{FominGS21},   use the  result of Marx~\cite{Marx08} about enumeration of subhypergraphs with certain properties of a given hypergraph of a special type. 
Basically, these algorithms can be seen as an intricate reduction of a clustering instance to a hypergraph of special type and then calling the result of Marx as a black box. 
In the context of the categorical clustering problems, a similar reduction  implies that all potential medians can be listed in  $2^{\Oh(B\log B)}|\Sigma|^B\cdot (mn)^{\Oh(1)}$  time (see Lemma~\ref{lem:enum-means}).

However, this strategy does not work to prove Theorem~\ref{thm:FPT-k}. Here the difficulties are due to the  constraints on sizes of clusters. 
 The algorithm for \probClust in~\cite{FominGP20} uses an observation that
 identical columns $\bfa_i$ and $\bfa_j$ of $\bfA$ can be clustered together.  That is, $i,j\in I_h$ for a cluster $I_h$ of an optimal solution.  Hence, a solution can be seen as a partition of the family of \emph{initial} clusters, i.e., inclusion maximal sets of indices $J\subseteq\{1,\ldots,n\}$ such that the columns $\bfa_i$ for $i\in J$ are the same. Since the number of initial clusters that are part of \emph{composite} clusters of a solution, that is, clusters including at least two initial clusters, is at most $2B$ in any yes-instance, the color coding technique of Alon,  Yuster and Zwick~\cite{AlonYZ95} allows to highlight initial clusters that may be included in a single composite cluster of a solution. This way, the initial problem is reduced to selecting a single composite cluster of minimum cost that contains a given number of initial clusters.  To solve this problem, the  result of Marx~\cite{Marx08} about enumeration of subhypergraphs becomes handy.

 This scheme does not work for \probCClust, because it may happen that splitting of an initial cluster between clusters of a solution is inevitable due to size constraints. This makes it impossible to select composite clusters independently from each other and destroys the approach used in~\cite{FominGP20,FominGS21}.

The main insight that allows to overcome the above issues is the very specific structure of  possible splitting of initial clusters (Lemma~\ref{lem:forest}). For a clustering $\mathcal{I}=\{I_1,\ldots,I_k\}$ and the partition $\mathcal{J}$ of the column indices into initial clusters, we look at the structure of the  
 intersection graph $G(\mathcal{I},\mathcal{J})$ defined by the two partitions of $\{1,\ldots,n\}$. The crucial fact we prove here is that there is an optimal solution such that this intersection graph is a forest. It can be seen that $G(\mathcal{I},\mathcal{J})$ has at most $3B$ vertices in connected components with at least three vertices for such a solution. This allows to guess the structure of 
$G(\mathcal{I},\mathcal{J})$, that is, guess a forest $F$ isomorphic to $G(\mathcal{I},\mathcal{J})$, by using the brute force. Then for a given $F$, we find a solution $\mathcal{I}$ with  $G(\mathcal{I},\mathcal{J})$ isomorphic to $F$ by combining dynamic programming  with color coding and enumeration of subhypergraphs of Marx. 

\paragraph*{Related work.}  
Clustering is one of the most common procedures in unsupervised  machine learning. 
\probCClust  is   the variant of the  popular $k$-median clustering with the Hamming norm.
In many applications of clustering,  constraints come naturally.  For example, the lower bound on the size of a cluster ensures certain anonymity of data and is often required for data privacy \cite{rosner2018privacy}. There is a rich literature on approximation algorithms for various versions of capacitated clustering \cite{Aggarwal, ByrkaRU16,ByrkaFRS15,CharikarGTS02,ChuzhoyR05,DemirciL16,Li17,DBLP:conf/icalp/Cohen-AddadL19,chen2016matroid,MalinenF14,Vallejo-HuangaM17}. 
However, to the best of our knowledge, no parameterized algorithms for categorical clustering with constraints on the sizes of clusters, were known prior to our work.

Several approximations and parameterized algorithms are known for the vanilla case of \probClust without constraints can be found in the literature. For binary field, \probClust
was introduced by Kleinberg, Papadimitriou, and Raghavan \cite{KleinbergPR04} as one of the examples of segmentation problems. The problem appears under different names in the literature~\cite{CilibrasiIK07,MiettinenMGDM08}.
Feige proved in \cite{Feige14b} that the problem is  \classNP-complete for every $k\geq 2$. We use several ideas from Feige's construction for our lower bounds.  
 Ostrovsky and Rabani \cite{OstrovskyR02} gave a randomized PTAS for binary \probClust which was recently improved to  EPTAS in 
 \cite{FominGLP020} and \cite{BanBBKLW19}. 
 Fomin, Golovach and Simonov in \cite{FominGS21} studied $k$-clustering with various distance norms in 
\probClust. One of their results is that clustering with Hamming-distance ($
\ell_0$-distance) (but unbounded size of the alphabet $\Sigma$) is $\classW1$-hard parameterized by $m+B$. 
 The following paper about binary variant of \probClust 
 is highly relevant to this paper. 
  Fomin, Golovach and Panolan \cite{FominGP20} gave two parameterized algorithms for 
binary case of \probClust with running time
$2^{\Oh(B\log B)}\cdot (nm)^{\Oh(1)}$ and $2^{\Oh(\sqrt{kB\log{(k+B)\log k}})} \cdot (nm)^{\Oh(1)}$.

\paragraph*{Organization of the paper.} In Section~\ref{sec:prelim}, we  introduce basic notions and notation used throughout the paper. We also show some auxiliary claims. In particular, we show that \probCClust is \classNP-complete for $k=2$ and binary matrices even if the clusters are required to be of the same size. In Section~\ref{sec:fpt}, we show our main result by constructing an \classFPT algorithm for \probCClust parmeterized by $B$. In Section~\ref{sec:variants}, we discuss  \probBClust and \probFBClust. We conclude in Section~\ref{sec:concl}, by discussing kernelization and stating some open problems.

\section{Preliminaries}\label{sec:prelim}
In this section we introduce the terminology used throughout the paper and obtain some auxiliary results. 

\subsection{Basic notation}

\paragraph{Matrices and vectors.}
All matrices and vectors considered in this paper are assumed to be over a finite alphabet $\Sigma$ and we say that a matrix (vector) is \emph{binary} if $\Sigma=\{0,1\}$. 
Therefore, to simplify notation, we omit $\Sigma$ in the notation whenever it does not create confusion. 
We use $m$ and $n$ to denote the number of rows and columns, respectively, of input matrices if it does not create confusion. 
We write $\bfA=(\bfa_1,\ldots,\bfa_n)$ to denote that $\bfA$ is a matrix with $n$ columns $\bfa_1,\ldots,\bfa_n$. 
For a partition $\mathcal{I}=\{I_1,\ldots,I_k\}$ of $\{1,\ldots,n\}$, we say that $\{I_1,\ldots,I_k\}$ is a \emph{$k$-clustering} for $\bfA$.
For an inclusion maximal $J\subseteq \{1,\ldots,n\}$ such that the columns $\bfa_i$ are identical for all $i\in J$, we say that $J$ is an \emph{initial cluster}. 
We say that a cluster $I_i$ of $\mathcal{I}$
is \emph{simple} if $I_i\subseteq J$ for some initial cluster $J$ and $I_i$ is \emph{composite}, otherwise, that is, if $I_i$ contains some $h,j\in\{1,\ldots,n\}$ such that $\bfa_h$ and $\bfa_j$ are distinct.
For a vector $\bfa\in\Sigma^m$, we use $\bfa[i]$ to denote the $i$-th element of the vector for $i\in\{1,\ldots,m\}$. Thus, for two  vectors $\bfa,\bfb\in\Sigma^m$, 
$\hdist(\bfa,\bfb)=|\{i\in\{1,\ldots,m\}\mid \bfa[i]\neq \bfb[i]\}|$. 
Let $a_{ij}$  for  $i\in\{1,\ldots,m\}$ and $j\in\{1,\ldots,n\}$ be  the elements of $\bfA$.
 For $I\subseteq\{1,\ldots,m\}$ and $J\subseteq\{1,\ldots,n\}$, we denote by $\bfA[I,J]$ the $|I|\times |J|$-submatrix of $\bfA$ with the elements $a_{ij}$ where $i\in I$ and $j\in J$. 

\paragraph{Parameterized complexity.} We refer to the books of Cygan et al.~\cite{CyganFKLMPPS15} and Downey and Fellows~\cite{DowneyF13}
for the detailed introduction to the field, see also the recent book of Fomin et al. on kernelization~\cite{FominLSZ19}. Here, we just informally sketch basic notions.

The input of a parameterized problem contains an integer value $k$ that is referred as a \emph{parameter}. 
A parameterized problem is \emph{fixed-parameter tractable} (\classFPT) if there is an algorithm solving it in $f(k)\cdot |I|^{\Oh(1)}$ time, where $I$ is an input, $k$ is a parameter, and $f(\cdot)$ is a computable function; 
 the parameterized complexity class \classFPT is composed by  fixed-parameter tractable problems.

A {\em{kernelization algorithm}}, or simply a {\em{kernel}}, for a parameterized problem $P$ is an  algorithm that, given an instance $(I,k)$ of $P$, in polynomial in~$|I|$ and~$k$ time returns an instance $(I',k')$ of $P$ such that (i) $(I,k)$ and $(I,'k')$ are equivalent, that is, $(I,k)$ is a yes-instance if and only if $(I',k')$ is a yes-instance, and (ii)
$|I'|+k'\leq g(k)$ for some computable function $g(k)$.  It is said that $g(\cdot)$ is the \emph{size} of a kernel; if  $g(\cdot)$ is a polynomial, then the kernel is polynomial.  
It is well-known that every \classFPT problem admits a kernel but, up to some reasonable complexity assumptions, there are \classFPT problems that have no polynomial kernels. The typical assumption is that  $\classNP\not\subseteq\classCoNP/{\rm poly}$ (see~\cite{FominLSZ19} for details).

\subsection{Solutions, clusters and medians}
 Formally, for \probClust and its variants, a solution is formed by  clusters $I_1,\ldots,I_k$ together with the corresponding medians $\bfc_1,\ldots,\bfc_k$. However, given clusters $I_1,\ldots,I_k$, optimal medians $\bfc_1,\ldots,\bfc_k$ can be computed by the easy \emph{majority rule}.  Let $\bfA=(\bfa_1,\ldots,\bfa_n)$ and let $\{I_1,\ldots,I_k\}$ be an $k$-clustering. For every $i\in\{1,\ldots,k\}$, we compute $\bfc_i\in\Sigma^m$ as follows. For each $j\in\{1,\ldots,m\}$, we consider the multiset 
 $R_{ij}=\{\bfa_h[j]\mid h\in I_i\}$ of elements of $\Sigma$. For each $s\in R_{ij}$, we compute the number of its occurrences  in the multiset and find an element $s^*$ that occurs most often (ties are broken arbitrarily). Then we set $\bfc_i[j]=s^*$. It is straightforward to verify that for every $\bfc\in \Sigma^m$, 
 $\sum_{h\in I_i}\hdist(\bfc_i,\bfa_h)\leq \sum_{h\in I_i}\hdist(\bfc,\bfa_h)$.  Therefore, the choice of $\bfc_i$ is optimal. This gives the following observation.
 
 \begin{observation}\label{obs:majority} 
 Given a matrix $\bfA=(\bfa_1,\ldots,\bfa_n)$ and a $k$-clustering $\{I_1,\ldots,I_k\}$, a family of vectors $\bfc_1,\ldots,\bfc_k\in \Sigma^m$ such that 
 $$\sum_{i=1}^k\sum_{j\in I_i}\hdist(\bfc_i,\bfa_j)$$
 is minimum can be computed in polynomial time by the majority rule.
  \end{observation}
   
 For a $k$-clustering   $\{I_1,\ldots,I_k\}$, we define the \emph{cost} $\cost(I_1,\ldots,I_k)$ as the minimum value of $\sum_{i=1}^k\sum_{j\in I_i}\hdist(\bfc_i,\bfa_j)$ over all $k$-tuples of vectors $\bfc_1,\ldots,\bfc_k\in\Sigma^m$. By Observation~\ref{obs:majority}, we have that  $\cost(I_1,\ldots,I_k)$ can be computed in polynomial time. Then the task of \probClust and its variants is reduced to finding a $k$-clustering of cost at most $B$ (with the respective constraints of the cluster sizes).  Thus, we may refer to a $k$-clustering as a solution without specifying medians. 
 
Observe that given vectors $\bfc_1,\ldots,\bfc_k$, we can find an $k$-clustering $\{I_1,\ldots,I_k\}$ that minimizes  $\sum_{i=1}^k\sum_{j\in I_i}\hdist(\bfc_i,\bfa_j)$ by the greedy procedure. For each $i\in\{1,\ldots,n\}$, we find $j\in\{1,\ldots,k\}$ such that $\hdist(\bfc_j,\bfa_i)$ is minimum (ties are broken arbitrarily) and place $i$ in the cluster $I_j$. Since 
$$\sum_{i=1}^n\min\{\hdist(\bfc_j,\bfa_i)\mid 1\leq j\leq k\}\leq \sum_{i=1}^k \sum_{j\in I_i}\hdist(\bfc_i,\bfa_j),$$
for every $k$-clustering $\{I_1,\ldots,I_k\}$, the described greedy procedure produces optimal partition of $\{1,\ldots n\}$ (some sets may be empty). However,   
the constructed $k$-clustering does not respect the size constraints of our problems. Still, given vectors  $\bfc_1,\ldots,\bfc_k$, we can decide in polynomial time whether an instance of
\probCClust has a solution with the medians $\bfc_1,\ldots,\bfc_k$ using a reduction to the classical \probPM problem on bipartite graphs that is well-known to be solvable in polynomial time by the Hungarian method of Kuhn~\cite{Kuhn55} (see also~\cite{LovaszP09}).

Recall that a \emph{matching} $M$ of a graph $G$ is a set of edges without common vertices. It is said that a matching $M$ \emph{saturates} a vertex $v$ if $M$
has an edge incident to $v$. A matching $M$ is \emph{perfect} if every vertex of $G$ is saturated. The task of \probPM is, given a bipartite graph $G$ and a weight function $w\colon E(G)\rightarrow \mathbb{Z}_{\geq 0}$, find a perfect matching $M$ (if it exists) such that its weight $w(M)=\sum_{e\in M}w(e)$ is minimum.

\begin{lemma}\label{lem:means-clusters}
Let $\bfc_1,\ldots,\bfc_k\in \Sigma^m$. For an instance of \probCClust,  
it can be decided in polynomial time whether the instance has a solution with the family of medians $\{\bfc_1,\ldots,\bfc_k\}$. 
\end{lemma}

\begin{proof}
Let $(\bfA,\Sigma,k,B,p,q)$ be an instance of \probCClust, $\bfA=(\bfa_1,\ldots,\bfa_n)$. Clearly, we can assume that the parameter $k$ in the instance equals to the number of vectors $\bfc_i$ as, otherwise, $\bfc_1,\ldots,\bfc_k$ cannot be the medians. We also assume without loss of generality that $p\leq \frac{n}{k}\leq q$; otherwise,
$(\bfA,\Sigma,k,B,p,q)$ is a no-instance.

We construct the bipartite graph $G$ as follows.
\begin{itemize}
\item For each $i\in\{1,\ldots,k\}$, construct a set of $p$ vertices $W_i=\{v_1^i,\ldots,v_p^i\}$ and a set of $q-p$ vertices $W_i'=\{v_{p+1}^i,\ldots,v_q^i\}$; note that $W_i'=\emptyset$ if $p=q$. Let $V_i=W_i\cup W_i'$ for $i\in\{1,\ldots,k\}$ and denote $V=\bigcup_{i=1}^kV_i$; the block of vertices $V_i$ corresponds to the median $c_i$.
\item For each $i\in\{1,\ldots,n\}$, construct a vertex $u_i$ corresponding to the column $\bfa_i$ of $\bfA$ and make $u_i$ adjacent to the vertices of $V$. Denote $U=\{u_1,\ldots,u_n\}$.
\item Construct a set of $s=kq-n$ vertices $U'=\{u_1',\ldots,u_s'\}$ that we call \emph{fillers} and make the vertices of $U'$ adjacent to the vertices of $W_j'$ for all $j\in\{1,\ldots,k\}$; note that $U'=\emptyset$ if $n=qk$ and observe that $qk\geq n$ by our assumption about the instance of \probCClust. 
\end{itemize}
Observe that $G$ is a bipartite graph, where $U\cup U'$ and $V$ form the bipartition. Note also that $|U\cup U'|=|V|=qk$. 

We define the edge weights as follows.
\begin{itemize}
\item For every $i\in\{1,\ldots,n\}$ and $j\in \{1,\ldots,k\}$, set $w(u_iv_h^j)=\hdist(\bfc_j,\bfa_i)$ for $h\in\{1,\ldots,q\}$, that is, the weight of all edges joining $u_i$ corresponding to $\bfa_i$ with the vertices of $V_j$ corresponding to the median $\bfc_j$ are the same and coincide with the Hamming distance between $\bfa_i$ and $\bfc_j$.
\item For every $i\in\{1,\ldots,s\}$ and $j\in \{1,\ldots,k\}$, set $w(u_i'v_h^j)=0$ for $h\in\{p+1,\ldots,q\}$, that is, the edges incident to the fillers have zero weights.
\end{itemize}

We now show that $G$ has a perfect matching of weight at most $B$ if and only if there is \linebreak a $k$-clustering 
 $\{I_1,\ldots,I_k\}$ for $A$ such that $p\leq |I_i|\leq q$ for all $i\in\{1,\ldots,k\}$ and \linebreak $\sum_{i=1}^k\sum_{j\in I_i}\hdist(\bfc_i,\bfa_j)\leq B$. 
 
 In the forward direction, assume $G$ has a perfect matching $M\subseteq E(G)$ of weight at most $B$. We construct the clustering $\{I_1,\ldots,I_k\}$ as follows. 
 For every $h \in \{1,\ldots,n\}$, $u_h$ is saturated by $M$ and, therefore, there are $i_h\in\{1,\ldots,k\}$ and $j_h \in \{1,\ldots q\}$ such that edge $u_hv^{i_h}_{j_h}\in M$. 
Consider $M'=\{u_hv^{i_h}_{j_h}\mid 1\leq h\leq n\}\subseteq M$. We cluster the columns of $\bfA$ according to $M'$. Formally, we place $h$ in $I_{i_h}$ for each $h\in\{1,\ldots,n\}$. Clearly, $\{I_1,\ldots,I_k\}$ is a partition of $\{1,\ldots,n\}$. 
Observe that for each $i\in\{1,\ldots,k\}$, the vertices of $W_i$ are adjacent only to the vertices of $U$. Since these vertices are saturated by $M$, we obtain that 
$|I_i|\geq p$ for every $i\in\{1,\ldots,k\}$. Since $|V_i|=q$, $|I_i|\leq q$ for all $i\in\{1,\ldots,k\}$. Now we upper bound the cost of the obtained $k$-clustering: 
\begin{equation*} 
 \sum_{i=1}^k\sum_{j\in I_i}\hdist(\bfc_i,\bfa_j)=\sum_{h=1}^n\hdist(\bfc_{i_h},\bfa_h)=w(M')\leq w(M)\leq B. 
\end{equation*}

For the reverse direction, consider a $k$-clustering $\{I_1,\ldots,I_k\}$  for $A$ such that $p\leq |I_i|\leq q$ for all $i\in\{1,\ldots,k\}$ and $\sum_{i=1}^k\sum_{j\in I_i}\hdist(\bfc_i,\bfa_j)\leq B$. Let $i \in \{1,\ldots,k\}$. 
Consider the cluster $I_i$ and assume that  $I_i=\{j_1,\ldots,j_{h_i}\}$. Recall that every vertex of  $V_i$ is adjacent to every vertex of $U$.  Denote by 
$M_i=\{u_{j_1}v_1^i,\ldots,u_{j_{h_i}}v_{h_i}^i\}$. Clearly, $M_i$ is a matching saturating the first $p\leq h_i\leq q$ vertices of $V_i$. In particular, the vertices of $W_i$ are saturated.
We construct $M_i$ for every $i\in\{1,\ldots,k\}$ and set  $M'=\bigcup_{i=1}^{k}M_i$. Since $\{I_1,\ldots,I_k\}$ is a partition of $\{1,\ldots,n\}$, $M'$ is a matching saturating every vertex of $U$. Denote by $V'$ the set of vertices of $V$ that are not saturated by $M'$. Notice that $V'\subseteq \bigcup_{i=1}^kW_i'$, because the vertices of each $W_i$ are saturated by $M_i$. Observe that every vertex of $U'$ is adjacent to every vertex of $W_i'$ for $i\in\{1,\ldots,k\}$, that is, $G[U'\cup V']$ is a complete bipartite graph.  
Because $|U'|=|V'|=s$, $G[U'\cup V']$ has a perfect matching $M''$. We set $M=M'\cup M''$. It is easy to see that $M$ is a matching and, since $M$ saturates every vertex of $G$, $M$ is a perfect matching. To evaluate the weight of $M$, recall that the edges of $G$ incident to the fillers have zero weights, that is, $w(M'')=0$.
Then 
\begin{align*}
w(M)=& w(M')=w(\bigcup_{i=1}^{k}M_i)=\sum_{i=1}^{k}\sum_{e \in M_i}w(e)=\sum_{i=1}^{k}\big(w(u_{j_1}v_1^i)+\cdots+w(u_{j_{h_i}}v_{h_i}^i)\big)\\
=&\sum_{i=1}^{k}\sum_{j \in I_i}^{}d_{\textbf{H}}(\bfc_i,\bfa_j) \leq B,
\end{align*}
and we conclude that $M$ is a perfect matching of weight at  most $B$.

It is straightforward to see that the construction of the graph $G$ from an instance $(\bfA,\Sigma,k,B,p,q)$ of \probCClust can be done in polynomial time. 
Then, because a perfect matching of minimum weight in $G$ can be found in polynomial time~\cite{Kuhn55,LovaszP09}, \probCClust can be solved in polynomial time.
This completes the proof of the lemma.
\end{proof}

By Lemma~\ref{lem:means-clusters}, we have that solving our problems can be reduced to finding 
a family of medians $\{\bfc_1,\ldots,\bfc_k\}$ (notice that some medians may be the same). 
 
 \subsection{Hardness of clustering}
 Since we are interested in the parameterized complexity of clustering problems, in the last part of this section, we argue that \probBClust, \probFBClust and  \probCClust are \classNP-hard for very restricted instances.
 
 In~\cite{Feige14b}, Feige proved that \probClust is \classNP-complete for $k=2$ and binary matrices, that is, for the case $\Sigma=\{0,1\}$. This result immediately implies that 
 \probCClust is also \classNP-complete for $k=2$ and binary matrices.  To see it,  note that 
 an instance $(\bfA,\Sigma,k,B)$ of \probClust is equivalent to the instance $(\bfA,\Sigma,k,B,p,q)$ of \probCClust for $p=1$ and $q=n$. However, we would like to underline that 
\probCClust is \classNP-hard even if  $p=q$.  
For this, we use some details of the hardness proof of Feige~\cite{Feige14b}.

 Feige proved that \probClust is \classNP-hard by showing a reduction from the \probMC problem~\cite{Feige14b}. In \probMC, we are given a graph $G$ and a nonnegative integer $\ell$,  and the task is to find a \emph{cut} $(S,\overline{S})$, that is, a partition of the vertex set into a set $S$ and its complement $\overline{S}=V(G)\setminus S$  such that the size of the cut, i.e, the number of edges between $S$ and $\overline{S}$ is at least $\ell$. The reduction constructed by Feige has the property given in the following lemma.
 
\begin{lemma}[\cite{Feige14b}]\label{lem:fiegeresult}
There is a polynomial time reduction from \probMC to \probClust that computes from an instance $(G,\ell)$ of \probMC an instance $(\bfA,\Sigma,2,B)$ of \probClust, where $\Sigma=\{0,1\}$, such that the following holds: if $(G,\ell)$ is a yes-instance of \probMC with a cut $(S,\bar{S})$ of size at least $\ell$, then $(\bfA,\Sigma,2,B)$ is a yes-instance of \probClust that has a solution $\{I_1,I_2\}$ with the property that $|I_1|/|I_2|=|S|/|\overline{S}|$.
\end{lemma}

\begin{theorem}\label{thm:NPhard}
For every fixed integer constant $c\geq 0$, \probCClust 
is \classNP-complete for $k=2$, binary matrices and $q-p\leq c$.
\end{theorem}

\begin{proof}
We show the lemma by the
reduction from \probMC that is well-known to be \classNP-complete~\cite{GareyJ79}. Given an instance $(G,\ell)$ of \probMC, we construct an auxiliary instance $(G',2\ell)$ of \probMC, where $G'$ is the union of two disjoint copies $G_1$ and $G_2$  of $G$. Then for the constructed instance $(G',2\ell)$  we can use as a black box the algorithm of Fiege \cite{Feige14b} from Lemma~\ref{lem:fiegeresult} to produce the instance $(\bfA,\Sigma,2,B)$ of \probClust
with $\Sigma=\{0,1\}$.  We further set $p=q=|V(G)|$ and consider the instance $(\bfA,\Sigma,2,B,p,q)$ of \probCClust. Clearly, $q-p\leq c$.
We show that $(G,\ell)$ is yes-instance of \probMC if and only if $(\bfA,\Sigma,2,B,p,q)$ is a yes-instance of \probCClust.

In the forward direction, assume that $(G,\ell)$ is a yes-instance of \probMC and let $(S,\bar{S})$ be a cut of size at least $\ell$. 
Let $S_1$ and $S_2$ be the copies of $S$ in $G_1$ and $G_2$, respectively.
We now consider $S'\subseteq V(G')$ such that $S'=S_1 \cup (V(G_2)\setminus S_2)$. Clearly, $\overline{S'}=(V(G_1)\setminus S_1)\cup S_2$ and $(S',\overline{S'})$ is a cut of $G'$ of size at least $2\ell$.  Moreover, $|S'|=|S_1|+|V(G_2)\setminus S_2|=|S_2|+|V(G_1)\setminus S_1|=|\overline{S'}|$. 
Hence, $(G',2\ell)$ is a yes-instance of \probMC with a solution $(S',\overline{S'})$ that has the property that $|S'|=|\overline{S'}|$. 
By Lemma~\ref{lem:fiegeresult}, $(\bfA,\Sigma,2,B)$ is a yes-instance of \probClust  that has a solution $\{I_1,I_2\}$ such that $|I_1|=|I_2|$.  
This implies that $p\leq |I_1|,|I_2|\leq q$. 
Therefore, $\{I_1,I_2\}$ is also solution for the instance $(\bfA,\Sigma,2,B,p,q)$ of \probCClust. Thus, $(\bfA,\Sigma,2,B,p,q)$ is a yes-instance of \probCClust.

In the reverse direction, suppose that $(\bfA,\Sigma,2,B,p,q)$ is a yes-instance of \probCClust. Then there is a 2-clustering $\{I_1,I_2\}$ for $\bfA$ of cost at most $B$.
This means that $(\bfA,\Sigma,2,B)$ is a yes-instance of \probClust. Because $(\bfA,\Sigma,2,B)$ is obtained from $(G',2\ell)$ by a polynomial reduction from Lemma~\ref{lem:fiegeresult},
$(G',2\ell)$ is a yes-instance of \probMC, that is, $G'$ has a cut of size at least $2\ell$. Since $G'$ is a disjoint union of two identical copies of $G$, each copy has a cut of size at least $l$. Therefore, $(G,\ell)$ is a yes-instance of \probMC. This completes the hardness proof.
\end{proof}  
 
 \section{\classFPT algorithm for parameterization by $B$}\label{sec:fpt}
 In this section we show that \probCClust is \classFPT when parameterized by $B$ and $|\Sigma|$. Our main result is Theorem~\ref{thm:FPT-k} that we restate here.

 \main*

 Note that this result is tight in the sense that it is unlikely that the dependence on the alphabet size could be made polynomial. It was shown in~\cite{FominGS21}, that \probClust is \classW{1}-hard when parameterized by $k$ and the number of rows $m$ of the input matrix if $\Sigma=\mathbb{Z}$, i.e., for an infinite alphabet. However, it is straightforward to see that this result holds for $\Sigma=\{0,\ldots,n-1\}$, because our measure is the Hamming distance. For each row of the input matrix, we can replace the original symbols by the symbols of  $\Sigma=\{0,\ldots,n-1\}$ in such a way that the original symbols in the row are the same if and only if the new symbols are the same. Clearly, this replacement gives an equivalent instance. This immediately leads to the following proposition. 
 
 \begin{proposition}\label{prop:w-hard} 
 \probCClust is  \classW{1}-hard when parameterized by $B$ and $m$. 
 \end{proposition}

 The remaining part of the section contains the proof of Theorem~\ref{thm:FPT-k}. The proof is constructive. In Subsection~\ref{subsec:tech}, we introduce some notation and show technical claims that are used by the algorithm, and Subsection~\ref{subsec:alg} contains the algorithm and its analysis. 

 \subsection{Technical lemmata}\label{subsec:tech}
 Let $\bfA=(\bfa_1,\ldots,\bfa_n)$ be a matrix.
 Recall that  an inclusion maximal $J\subseteq \{1,\ldots,n\}$ such that the columns $\bfa_i$ are identical for all $i\in J$ is called an initial cluster.
 Suppose that $\{I_1,\ldots,I_k\}$ is a $k$-clustering for $\bfA$. Recall that  a cluster $I_i$ is simple if $I_i\subseteq J$ for some initial cluster $J$ and $I_i$ is composite, otherwise, that is, if $I_i$ contains some $h,j\in\{1,\ldots,n\}$ such that $\bfa_h$ and $\bfa_j$ are distinct.
 
 We start by making the following observation about medians of sufficiently big (in $B$) clusters.

\begin{observation}\label{obs:big-clust}
Let $\{I_1,\ldots,I_k\}$ be an $k$-clustering for a matrix $\bfA=(\bfa_1,\ldots,\bfa_n)$ of cost at most $B$, and let $|I_i|\geq B+1$ for some $i\in\{1,\ldots,k\}$.
Then for all vectors $\bfc_1,\ldots,\bfc_k\in\Sigma^m$ such that $\sum_{h=1}^k\sum_{j\in I_h}\hdist(\bfc_h,\bfa_j)\leq B$, $\bfc_i=\bfa_j$ for at least $|I_i|-B$ indices $j\in I_i$.
Moreover, if $|I_i|\geq 2B+1$, then $\bfc_i$ is unique. 
\end{observation}
 
\begin{proof}
To show the first part of the claim, assume that $\bfc_i\in \Sigma^m$ is distinct from at least $B+1$ columns $\bfa_j$ for $j\in I_i$. Then
 \begin{equation*}
 B\geq \sum_{h=1}^k\sum_{j\in I_h}\hdist(\bfc_h,\bfa_j)\geq \sum_{j\in I_i}\hdist(\bfc_i,\bfa_j)\geq B+1;
  \end{equation*}
  a contradiction. For the second part of the claim, note that if $|I_i|\geq 2B+1$, then $\bfc_i$ should coincide with more than half columns $\bfa_j$ with $j\in I_i$ and, therefore, the choice of $\bfc_i$ is unique.
 \end{proof}  
 
We use the following simple observation about the number of composite clusters and the number of  initial cluster having elements in the composite clusters of a solution. 
 
\begin{observation}\label{obs:composite} 
 Let $\bfA=(\bfa_1,\ldots,\bfa_n)$ be a matrix with the partition $\mathcal{J}=\{J_1,\ldots,J_s\}$ of $\{1,\ldots,n\}$ into initial clusters. Let also 
$\mathcal{I}=\{I_1,\ldots,I_k\}$ be a $k$-clustering for $\bfA$ of cost at most $B$. Then $\mathcal{I}$  contains at most $B$ composite clusters and  $\mathcal{J}$ 
has at most $2B$ initial clusters with nonempty intersections with the composite clusters of $\mathcal{I}$.
\end{observation} 
 
 \begin{proof}
 Let $\bfc_1,\ldots,\bfc_k$ be medians such that $\sum_{i=1}^k\sum_{j\in I_i}\hdist(\bfc_i,\bfa_j)\leq B$. Note that if $I_i$ is a composite cluster for some $i\in\{1,\ldots,k\}$, then $\bfc_i$ is distinct from $\bfa_j$ for at least one $j\in I_i$ and $\sum_{j\in I_i}\hdist(\bfc_i,\bfa_j)\geq 1$. Therefore, $\mathcal{I}$  contains at most $B$ composite clusters.   For the second claim, notice that if $\mathcal{J}$ 
has $t\geq B$ initial clusters with nonempty intersections with composite clusters, then because $\mathcal{I}$ has at most $B$ composite clusters, for at least $t-B$ these initial clusters $J_j$, $\bfa_h\neq \bfc_i$ for $h\in J_j$ and all the medians $\bfc_i$ of composite clusters.  Hence, $t\leq 2B$.
  \end{proof}

Let  $J\subseteq \{1,\ldots,n\}$  be an initial cluster. Due to size constraints, it may happen that  a $k$-clustering  $\{I_1,\ldots,I_k\}$ with several simple clusters $I_i\subseteq J$ provides a solution. This means, that we should partition a subset of $J$ into blocks of bounded size. To verify whether we are able to create such a partition, we use the following observation. 

\begin{observation}\label{obs:blocks}
Let $p$ and $q$ be positive integers, $p\leq q$. A finite set $X$ can be partitioned into $h$ subsets such that each of them has size at least $p$ and at most  $q$ if and only if $\Big\lceil\frac{|X|}{q}\Big\rceil\leq h\leq \Big\lfloor\frac{|X|}{p}\Big\rfloor$.
\end{observation}

\begin{proof}
If $X$ can be partitioned into $h$ subsets of size at least $p$ and at most  $q$, then, trivially, $ph\leq |X|$ and $qh\geq |X|$, i.e., 
$\Big\lceil\frac{|X|}{q}\Big\rceil\leq h\leq \Big\lfloor\frac{|X|}{p}\Big\rfloor$.
If $\Big\lceil\frac{|X|}{q}\Big\rceil\leq h\leq \Big\lfloor\frac{|X|}{p}\Big\rfloor$, then $X$ has $h$ disjoint subsets $X_1,\ldots,X_h$ of size $p$. Then the remaining $|X|-ph$ elements can be greedily added to these subsets without exceeding the upper bound $q$ on the size.
\end{proof}

Let $\mathcal{J}=\{J_1,\ldots,J_s\}$ be the partition of $\{1,\ldots,n\}$ into initial clusters. For a $k$-clustering $\mathcal{I}=\{I_1,\ldots,I_k\}$, we define the graph $G(\mathcal{I},\mathcal{J})$ as the intersection graph of the sets of $\mathcal{I}$ and $\mathcal{J}$, that is,  $G(\mathcal{I},\mathcal{J})$ is the bipartite graph with the set of vertices $\mathcal{I}\cup \mathcal{J}$ such that for every $i\in\{1,\ldots,k\}$ and $j\in\{1,\ldots,s\}$, 
$I_i$ and $J_j$ are  adjacent if and only if $I_i\cap J_j\neq\emptyset$. We show that we can assume $G(\mathcal{I},\mathcal{J})$ to be a forest. This can be proved using an ILP or flow formulation of the clustering problem with given medians. For simplicity, we provide a direct proof.   
   
\begin{lemma}\label{lem:forest}
Let $\bfA=(\bfa_1,\ldots,\bfa_n)$ be a matrix with the partition $\mathcal{J}=\{J_1,\ldots,J_s\}$ of $\{1,\ldots,n\}$ into initial clusters. Let also 
$\mathcal{I}=\{I_1,\ldots,I_k\}$ be a $k$-clustering for $\bfA$. Then there is a $k$-clustering $\mathcal{I}'=\{I_1',\ldots,I_k'\}$
such that (i) $|I_i|=|I_i'|$ for all $i\in \{1,\ldots,k\}$, (ii) $\cost(I_1',\ldots,I_k')\leq \cost(I_1,\ldots,I_k)$, and (iii) $G(\mathcal{I}',\mathcal{J})$ is a forest.
\end{lemma}   
 
\begin{proof} 
Assume that $\mathcal{I}'=\{I_1',\ldots,I_k'\}$ is a $k$-clustering for $\bfA$ satisfying conditions (i) and (ii) such that the number of edges of  $G(\mathcal{I}',\mathcal{J})$ is minimum. Denote by $\bfc_1,\ldots,\bfc_k$ optimal medians for $I_1',\ldots,I_k'$. We claim that $G(\mathcal{I}',\mathcal{J})$ is a forest.  

The proof is by contradiction. Assume that $G(\mathcal{I}',\mathcal{J})$ has a cycle. This means that there are distinct $i_1,\ldots,i_t\in\{1,\ldots k\}$ and distinct $j_1,\ldots,j_t\in\{1,\ldots,s\}$ such that $I_{i_h}'\cap J_{j_h}\neq\emptyset$ and $I_{i_h}'\cap J_{j_{h+1}}\neq\emptyset$ for all $h\in\{1,\ldots,t\}$; here and further in the proof, we assume that $j_{t+1}=j_1$ and $i_{t+1}=i_1$. 

For $h\in\{1,\ldots,s\}$, denote by $\bfb_h$ the vector coinciding with $\bfa_{h'}$ for $h'\in J_h$. We observe that either
\begin{equation}\label{eq:case-one}
\sum_{h=1}^t(\hdist(\bfc_{i_h},\bfb_{j_h})+\hdist(\bfc_{i_{h}},\bfb_{j_{h+1}}))\geq 2\sum_{h=1}^t\hdist(\bfc_{i_h},\bfb_{j_h})
\end{equation}
 or
 \begin{equation}\label{eq:case-two}
\sum_{h=1}^t(\hdist(\bfc_{i_h},\bfb_{j_h})+\hdist(\bfc_{i_{h}},\bfb_{j_{h+1}}))\geq 2\sum_{h=1}^t\hdist(\bfc_{i_h},\bfb_{j_{h+1}}),
\end{equation}
because the sums of the left and right parts of inequalities (\ref{eq:case-one}) and (\ref{eq:case-two}) are the same.

We assume without loss of generality that (\ref{eq:case-one}) holds, as the second case is symmetric. This means that 
\begin{equation}\label{eq:shift}
\sum_{h=1}^t\hdist(\bfc_{i_{h}},\bfb_{j_{h+1}})\geq \sum_{h=1}^t\hdist(\bfc_{i_h},\bfb_{j_h}).
\end{equation}
We iteratively modify $\mathcal{I}'$ by moving a representative of $J_{j_h}$ in $I_{i_{h-1}}$ to $I_{i_h}$ for $h\in\{2,\ldots,t+1\}$, that is, representatives are moved cyclically without changing the cluster sizes. We show that this  procedure does not increase the clustering cost with respect to the medians $\bfc_1,\ldots,\bfc_k$. 

Formally, we construct the $k$-clusterings $\mathcal{I}^{(0)},\mathcal{I}^{(1)},\ldots$,
where $\mathcal{I}^{(p)}=\{I_1^{(p)},\ldots,I_k^{(p)}\}$ for $p=0,1,\ldots$, starting from $\mathcal{I}^{(0)}=\mathcal{I}'$ while 
 $J_{j_{h+1}}\cap I_{i_{h}}^{(p)}\neq\emptyset$ for all $h\in\{1,\ldots,t\}$.

Assume that $\mathcal{I}^{(p)}=\{I_1^{(p)},\ldots,I_k^{(p)}\}$ is constructed and  $J_{j_{h+1}}\cap I_{i_{h}}^{(p)}\neq\emptyset$ for all $h\in\{1,\ldots,t\}$.
For every $h\in\{1,\ldots,t\}$, let $i_h'\in J_{j_{h+1}}\cap I_{i_{h}}^{(p)}$. We define $\mathcal{I}^{(p+1)}=\{I_1^{(p+1)},\ldots,I_k^{(p+1)}\}$ by setting 
\begin{equation*}
I_{i_{h}}^{(p+1)}=(I_{i_{h}}^{(p)}\setminus\{i_{h}'\})\cup\{i_{h-1}'\}
\end{equation*}
for all $h\in\{1,\ldots,t\}$ assuming that $i_{0}'=i_t'$, and we set $I_q^{(p+1)}=I_q^{(p)}$ for $q\in\{1,\ldots,k\}\setminus\{i_1,\ldots,i_t\}$. Clearly,
$|I_i^{(p+1)}|=|I_q^{(p)}|$ for all $i\in\{1,\ldots,r\}$.
We have that
\begin{align*}
\big(\sum_{i=1}^{k}\sum_{j\in I_i^{(p)}}\hdist(\bfc_i,\bfa_j)\big)-\big(\sum_{i=1}^{k}\sum_{j\in I_i^{(p+1)}}\hdist(\bfc_i,\bfa_j)\big)=&
\sum_{h=1}^t(\hdist(\bfc_{i_h},\bfa_{i_{h}'})-\hdist(\bfc_{i_h},\bfa_{i_{h-1}'}))\\
=&\sum_{h=1}^t(\hdist(\bfc_{i_h},\bfb_{j_{h+1}})-\hdist(\bfc_{i_h},\bfb_{j_h}))\\
=&\big(\sum_{h=1}^t(\hdist(\bfc_{i_h},\bfb_{j_{h+1}})\big)-\big(\sum_{h=1}^t\hdist(\bfc_{i_h},\bfb_{j_h})\big)\geq 0,
\end{align*}
where the last inequality follows from (\ref{eq:shift}). This means that the cost of the $k$-clustering $\mathcal{I}^{(p+1)}$ with respect to the medians $\bfc_1,\ldots,\bfc_k$ is at most the cost of $\mathcal{I}^{(p)}$ with respect to the same medians.

The next $k$-clustering $\mathcal{I}^{(p+1)}$ is constructed from $\mathcal{I}^{(p)}$ if $J_{j_{h+1}}\cap I_{i_{h}}^{(p)}\neq\emptyset$ for all $h\in\{1,\ldots,t\}$.
Thus, the sequence is finite, and for the last $k$-clustering $\mathcal{I}^{(q)}$, there is $h\in\{1,\ldots,t\}$
such that  $J_{j_{h+1}}\cap I_{i_{h}}^{(q)}=\emptyset$, that is, $I_{i_{h}}^{(q)}$ and 
$J_{j_{h+1}}$ are not adjacent in $G(\mathcal{I}^{(q)},\mathcal{J})$. Note that the rearrangement elements of clusters does not create new adjacencies in 
$G(\mathcal{I}^{(q)},\mathcal{J})$, because no cluster gets representatives of an initial cluster that have had no representatives in it.
We conclude that $G(\mathcal{I}^{(q)},\mathcal{J})$ has less edges than $G(\mathcal{I}',\mathcal{J})$ but this contradicts the choice of $\mathcal{I}'$.
Therefore,  $G(\mathcal{I}',\mathcal{J})$ is a forest and $\mathcal{I}'$ satisfies conditions (i)--(iii) as required.
\end{proof}
 
 Next, we show that, given a matrix $\bfA$, we can list all potential medians for a $k$-clustering of cost at most $B$ in \classFPT in $B$ time. We show this by making use of the nontrivial result of Marx~\cite{Marx08} about the enumeration of subhypergraphs  with bounded partial edge cover. This result already proved to be very useful for designing \classFPT algorithms for clustering problems~\cite{FominGP20,FominGS21}. 
 
 Recall that a hypergraph $\mathcal{H}$ is a pair $(V,\mathcal{E})$, where $V$ is a set of \emph{vertices} and $\mathcal{E}$ is a family of subsets of $V$ called \emph{hyperedges}. Similarly to graphs, we denote by $V(\mathcal{H})$ the set of vertices and by $\mathcal{E}(\mathcal{H})$ the set of hyperedges. For a vertex $v$, we denote by $\mathcal{E}_{\mathcal{H}}(v)$ the set of hyperedges containing $v$, that is, $\mathcal{E}_{\mathcal{H}}(v)=\{E\in \mathcal{E}(\mathcal{H})\mid v\in E\}$.
 
 Let $\mathcal{G}$ be a hypergraph and let $U\subseteq V(\mathcal{G})$. It is said that a hypergraph $\mathcal{H}$ \emph{appears at $U$ as a subhypergraph} if there is a bijection $\pi\colon V(\mathcal{H})\rightarrow U$ with the property that for every $E\in \mathcal{E}(\mathcal{H})$, there is $E'\in\mathcal{E}(\mathcal{G})$ such that 
 $\pi(E)=E'\cap U$.
 
A \emph{fractional hyperedge cover} of a hypergraph $\mathcal{H}$ is a function $\varphi\colon \mathcal{E}(\mathcal{H})  \rightarrow [0,1]$ such that for every vertex $v\in V(\mathcal{H})$, $\sum_{E\in\mathcal{E}_{\mathcal{H}}(v)}\varphi(E)\geq 1$, that is, the sum of the values assigned by $f$ of the hyperedges containing $v$ is at least one. The \emph{fractional cover number} ${\rho}^*(\mathcal{H})$ of $\mathcal{H}$ is the minimum value $\sum_{E\in \mathcal{E}(\mathcal{H})} \varphi(E)$ taken over all fractional hyperedge covers $\varphi$ of $H$. 

\begin{proposition}[\cite{Marx08}]\label{prop:hypergraph}
Let $\mathcal{H}$ be a hypergraph with fractional cover number
${\rho}^*(\mathcal{H})$, and let $\mathcal{G}$ be a hypergraph whose hyperedges have  size at most $\ell$. There
is an algorithm that enumerates, in  
$|V(\mathcal{H})|^{\Oh(|V(\mathcal{H})|)}\cdot \ell^{|V(\mathcal{H})|\rho^*(\mathcal{H})+1}\cdot |\mathcal{E}(\mathcal{G})|^{\rho^*(\mathcal{H})+1}\cdot |V(\mathcal{G})|^2     $  time, every  $U\subseteq V(\mathcal{G})$ where $\mathcal{H}$ appears at $U$ as subhypergraph  in $\mathcal{G}$. 
\end{proposition}

We apply this result similarly to~\cite{FominGS21} and, therefore, only briefly sketch the proof of the following lemma. 
 
\begin{lemma}\label{lem:enum-means} 
 There is an algorithm that, given a matrix $\bfA=(\bfa_1,\ldots,\bfa_n)$ and a nonnegative integer $B$, in $2^{\Oh(B\log B)}\cdot |\Sigma|^B\cdot (mn)^{\Oh(1)}$ time outputs a set $\mathcal{M}(\bfA,B)\subseteq \Sigma^m$ of size $2^{\Oh(B\log B)}\cdot |\Sigma|^B\cdot (mn)^{\Oh(1)}$ such that for every $k$-clustering $\{I_1,\ldots,I_k\}$ for $\bfA$ of cost at most $B$, there are $\bfc_1,\ldots,\bfc_k\in \mathcal{M}(\bfA,B)$ such that 
 $\sum_{i=1}^k\sum_{j\in I_i}\hdist(\bfc_i,\bfa_j)\leq B$. 
\end{lemma}

\begin{proof} 
Let $\mathcal{S}$ be the set of distinct columns of $\bfA$. Initially, we set $\mathcal{M}(\bfA,B):=\mathcal{S}$.

For every $\bfs\in\mathcal{S}$,
we construct the hypergraph $\mathcal{G}_\bfs$  with the vertex set $\{1,\ldots,m\}$ with 
hyperedges corresponding to the columns of $\bfA$ at Hamming distance at most $B$ from $\bfs$: for every $i\in\{1,\ldots,n\}$ such that $\hdist(\bfs,\bfa_i)\leq B$, we introduce the hyperedge
\begin{equation*}
E_i=\{j\mid 1\leq j\leq m\text{ and }\bfa_i[j]\neq \bfs[j]\}, 
\end{equation*}
that is, the hyperedge contains indices, where $\bfs$ differs from $\bfa_i$. Note that $|E_i|\leq B$.

Consider an arbitrary  $k$-clustering $\{I_1,\ldots,I_k\}$ for $\bfA$ of cost at most $B$. Let $i\in \{1,\ldots,k\}$ and let $\bfs\in \mathcal{S}$ be such that $\bfs=\bfa_j$ for some $j\in I_i$. Let also $\bfc_i\in\Sigma^m$ be an optimal median for $I_i$, that is,
$\sum_{j\in I_i}\hdist(\bfc_i,\bfa^j)$ is minimum. 
 Notice that if $|I_i|\geq B+1$, then by Observation~\ref{obs:big-clust}, every feasible medians for $I_i$ is a column of $\bfA$ and these columns are already placed in $\mathcal{M}(\bfA,B)$. Also if $\bfc_i=\bfs$, then $\bfc_i\in \mathcal{M}(\bfA,B)$.
Assume that $|I_i|\leq B$ and $\bfc_i\neq\bfs$.
 Clearly, $\hdist(\bfc_i,\bfs)\leq B$. Moreover, for any $j\in I_i$, 
$\hdist(\bfs,\bfa_j)\leq B$. This holds trivially if $\bfs=\bfa_j$. Otherwise, if $\bfs\neq \bfa_j$, we have that 
$\hdist(\bfs,\bfa_j)\leq \hdist(\bfc_i,\bfs)+\hdist(\bfc_i,\bfa_j)\leq \sum_{j\in I_i}\hdist(\bfc_i,\bfa_j)\leq B$.
Let 
\begin{equation*}
D=\{j\mid 1\leq j\leq m\text{ and }\bfc_i[j]\neq \bfs[j]\},
\end{equation*}
that is, $D$ is the set of indices where $\bfs$ differs from the median $\bfc_i$.

We consider the hypergraph $\mathcal{H}_i$ with the vertex set $D$ whose edges correspond to the columns $\bfa_j$ for $j\in I_i$. For each $j\in I_i$, we  construct the hyperedge 
\begin{equation*}
F_j=\{h\mid h\in D\text{ and }\bfa_j[h]\neq \bfs[h]\}, 
\end{equation*}
that is, each hyperedge contains indices from $D$, where $\bfs$ differs from $\bfa_j$. We claim that the fractional cover number $\rho^*(\mathcal{H}_i)\leq 2$.

To show this, we define the function $\varphi(F)=\frac{2}{|\mathcal{E}(\mathcal{H}_i)|}$ for every hyperedge $F$ of $\mathcal{H}_i$. We prove that $\varphi$ is a fractional hyperedge cover of $\mathcal{H}_i$. Thus, we have to show that for every $j\in D$,  $\sum_{F\in\mathcal{E}_{\mathcal{H}_i}(j)}\varphi(F)\geq 1$. This is equivalent to proving that for every $j\in D$, at least half of the hyperedges of $\mathcal{H}_i$ contain $j$. Assume that this is not the case, i.e., there is $j\in D$ such that more than half of hyperedges do not contain $j$. This mean that for more than half of columns $\bfa_h$ for $h\in I_i$, $\bfs[j]=\bfa_h[j]=s$. However, by the definition of $D$, $\bfs[j]\neq \bfc_i[j]$ and, therefore, $\bfc_i[j]\neq s$. This contradicts the assumption that $\bfc_i$ is an optimal median for $I_i$, because replacing the current value $\bfc_i[j]$ by $s$ decreases the cost. Hence, $\varphi$ is a fractional hyperedge cover. Then 
\begin{equation*}
\rho^*(\mathcal{H}_i)\leq \sum_{F\in \mathcal{E}(\mathcal{H}_i)}\varphi(F)=\sum_{F\in \mathcal{E}(\mathcal{H}_i)}\frac{2}{|\mathcal{E}(\mathcal{H}_i)|}=2.
\end{equation*}

Observe that $\mathcal{H}_i$ appears in $\mathcal{G}_\bfs$ at $D$, because for each $j\in I_i$, $\hdist(\bfs,\bfa_j)\leq B$, that is, for every $j\in I_i$, $\mathcal{G}_\bfs$ contains the hyperedge $E_j$ corresponding to $\bfa_j$; the mapping $\pi\colon V(\mathcal{H}_i)\rightarrow D$ is the identity function. 

We obtain that $\mathcal{H}_i$ is a hypergraph with the fractional cover number at most 2 that appears in   $\mathcal{G}_\bfs$ at $D$. Notice that, given $\bfs$ and $D$, we can list the vectors over $\Sigma^m$ that differ from $\bfs$ in the indices from $D$ and the total number of such vectors is at most $|\Sigma|^B$, because $|D|\leq B$.
Then $\bfc_i$ appears in this list.
This leads to the following algorithm. We consider all hypergraphs $\mathcal{H}$ on at most $B$ vertices with at most $B$ hyperedges. Then for each $\mathcal{H}$ and every $\bfs\in\mathcal{S}$, we use the algorithm of Marx from Proposition~\ref{prop:hypergraph} to  enumerate every  $D\subseteq V(\mathcal{G}_\bfs)$ where $\mathcal{H}$ appears in $\mathcal{G}_D$ as subhypergraph. Then for every $D$, we list  the vectors that differ from $\bfs$ in the indices from $D$ by brute force. Then these vectors are included in $\mathcal{M}(\bfA,B)$.

 For given $\mathcal{H}$ and $\bfs$, the sets $D$ can be enumerated in time 
$2^{\Oh(B\log B)}\cdot B^{2B+1}\cdot n^{3}\cdot m^2$ by Proposition~\ref{prop:hypergraph}. Then generating the vectors that differ from $\bfs$ in $D$ can be done in 
$|\Sigma|^B\cdot n^{\Oh(1)}$ time as we can assume that $|\Sigma|\leq n$. However, we need $2^{\Oh(B^2)}$ time to generate all hypergraphs with at most $B$ vertices and at most $B$ hyperedges. This gives the total running time $2^{\Oh(B^2)}\cdot|\Sigma|^B\cdot (mn)^{\Oh(1)}$ and the same bound on the size of $\mathcal{M}(\bfA,B)$.

The running time can be improved by proving that there is a subhypergraph $\mathcal{H}_i'$ of $\mathcal{H}$ with $V(\mathcal{H}_i')=V(\mathcal{H}_i)$ and  
$\mathcal{E}(\mathcal{H}_i')\subseteq \mathcal{E}(\mathcal{H}_i)$ of size $\Oh(\log B)$ (more precisely, of size at most $160\ln B$) such that $\rho^*(\mathcal{H}_i')\leq 4$. The proof is identical to the proof of Claim~18 of~\cite{FominGS21} (see also Proposition~6.3 of \cite{Marx08}) and we omit it here. 

Then we consider all hypergraphs  $\mathcal{H}$ with at most $B$ vertices and at most $160\ln B$ hyperedges. The total number of these hypergraphs is $2^{\Oh(B\log B)}$. Then, in the same way as above, for each $\mathcal{H}$ and every $\bfs\in\mathcal{S}$, we use the algorithm of Marx from Proposition~\ref{prop:hypergraph} to  enumerate every  $D\subseteq V(\mathcal{G}_\bfs)$ where $\mathcal{H}$ appears in $\mathcal{G}_D$ as subhypergraph. 
For  every $D$, the vectors that differ from $\bfs$ in the indices from $D$ are enumerate by  brute force and each vector is added to $\mathcal{M}(\bfA,B)$, unless it is already included in the set. The total running time is $2^{\Oh(B\log B)}\cdot |\Sigma|^B\cdot (mn)^{\Oh(1)}$ and the number of vectors in $\mathcal{M}(\bfA,B)$ is  $2^{\Oh(B\log B)}\cdot |\Sigma|^B\cdot (mn)^{\Oh(1)}$. 
\end{proof}

\subsection{Algorithm}\label{subsec:alg} 
Let $(\bfA,\Sigma,k,B,p,q)$ be an instance of \probCClust with $\bfA=(\bfa_1,\ldots,\bfa_n)$. 
First, we compute the partition $\mathcal{J}=\{J_1,\ldots,J_s\}$
of $\{1,\ldots,n\}$ into initial clusters. 

By Lemma~\ref{lem:forest}, if $(\bfA,\Sigma,k,B,p,q)$ is a yes-instance, then there is a solution $\mathcal{I}=\{I_1,\ldots,I_k\}$ to the instance such that
the intersection graph  $G(\mathcal{I},\mathcal{J})$ of the initial clusters and the clusters of the solution is a forest. We call such a solution (or $k$-clustering) \emph{acyclic}. To solve the problem, we check whether the considered instance has an acyclic solution. To simplify notation, we assume that all solution considered further on are acyclic.

By Observation~\ref{obs:composite}, any  $k$-clustering  for $\bfA$  of cost at most $B$ has at most $B$ composite clusters. We consecutively  consider $t=0,\ldots,\min\{B,k\}$, and for each $t$, we verify whether there is a solution  $\mathcal{I}=\{I_1,\ldots,I_k\}$ with exactly $t$ composite clusters. If we find such a solution, then we return the yes-answer and stop. Otherwise, if we have no solution for all the values of $t$, we report that $(\bfA,\Sigma,k,B,p,q)$ is a no-instance. From now, we assume that nonnegative $t\leq\min\{B,k\}$ is fixed.

It is convenient to consider the special case $t=0$ separately. If $t=0$, then a solution $\mathcal{I}$ has no composite cluster, that is, the clusters of the solution form partitions of the initial clusters. Observe that $\cost(\mathcal{I})=0\leq B$ in this case. 
By Observation~ \ref{obs:blocks}, the initial clusters can be partitioned into $k$ blocks of size at least $p$ and at most $q$, if and only if there are positive integers $h_1,\ldots,h_s$ such that $k=h_1+\cdots+h_s$ and 
$\Big\lceil\frac{|J_i|}{q}\Big\rceil\leq h_i\leq \Big\lfloor\frac{|J_i|}{p}\Big\rfloor$ for every $i\in\{1,\ldots,s\}$. For every $i\in\{1,\ldots,s\}$, we verify whether 
 $\Big\lceil\frac{|J_i|}{q}\Big\rceil\leq \Big\lfloor\frac{|J_i|}{p}\Big\rfloor$.  If at least one of the inequalities does not hold, the required $h_1,\ldots,h_s$ do not exist. Otherwise, we observe that positive integers $h_1,\ldots,h_s$ such that $k=h_1+\cdots+h_s$ and 
$\Big\lceil\frac{|J_i|}{q}\Big\rceil\leq h_i\leq \Big\lfloor\frac{|J_i|}{p}\Big\rfloor$ for every $i\in\{1,\ldots,s\}$ exist if and only if
$\sum_{i=1}^s\Big\lceil\frac{|J_i|}{q}\Big\rceil\leq k\leq \sum_{i=1}^s \Big\lfloor\frac{|J_i|}{p}\Big\rfloor$.
Then we verify the last inequality.

From now, we assume that $t\geq 1$. Note that we also can assume that $B\geq 1$, because for $B=0$, no cluster of a solution can be composite.

By Observation~\ref{obs:composite}, there are at most $2B$ initial clusters with nonempty intersections with the composite clusters of a solution $\mathcal{I}$. Since $G(\mathcal{I},\mathcal{J})$ is a forest, it is easy to observe that at least $t+1$ initial clusters have nonempty intersections with the composite clusters. We consider $\ell=t+1,\ldots,2B$, and for each $\ell$, we check whether  there is a solution  $\mathcal{I}=\{I_1,\ldots,I_k\}$ such that exactly $\ell$ initial clusters have nonempty intersections with the composite clusters of $\mathcal{I}$. If we find such a solution, then we return the yes-answer and stop. Otherwise, if we have no solution for all the values of $\ell$, we report that $(\bfA,\Sigma,k,B,p,q)$ is a no-instance. 
From now, we assume that positive $t+1\leq \ell\leq 2B$ is given.

 We use Lemma~\ref{lem:enum-means} to construct the set $\mathcal{M}=\mathcal{M}(\bfA,B)$ of potential medians. Recall that this set has size $2^{\Oh(B\log B)}|\Sigma|^B\cdot (mn)^{\Oh(1)}$ and can be computed in $2^{\Oh(B\log B)}|\Sigma|^B\cdot (mn)^{\Oh(1)}$ time. 
 For a $k$-clustering $\mathcal{I}=\{I_1,\ldots,I_k\}$, we define the \emph{minimum cost (with respect to $\mathcal{M}$)}, as 
 \begin{equation*}
 \min\{\sum_{i=1}^k\sum_{j\in I_i}\hdist(\bfc_i,\bfa_j)\mid \bfc_1,\ldots,\bfc_k\in \mathcal{M}\}.
 \end{equation*}
 If $(\bfA,\Sigma,k,B,p,q)$ is a yes-instance, then it has a solution  such that the medians are in $\mathcal{M}$ by Lemma~\ref{lem:enum-means}. 
Therefore, solving the problem is equivalent to finding a clustering of minimum cost at most $B$ with respect to $\mathcal{M}$. Throughout this section, whenever we say that $\mathcal{I}$ is a clustering of minimum cost, we mean that the cost is minimum with respect to $\mathcal{M}$. 

We use the \emph{color coding} technique of Alon, Yuster and Zwick~\cite{AlonYZ95} (see~\cite[Chapter~5]{CyganFKLMPPS15} for the detailed introduction). We first give a Monte Carlo algorithm with false negatives and then explain how to derandomize it. The main idea is to highlight the initial clusters with nonempty intersections with clusters of a potential solution. We color the initial clusters by $\ell$ colors uniformly at random. We say that a $k$-clustering $\mathcal{I}=\{I_1,\ldots,I_k\}$ of cost at most $B$ is a \emph{colorful} solution if the initial clusters with nonempty intersections  with the clusters of $\mathcal{I}$ have distinct colors. As it is standard for color coding, the algorithm exploits the property that if there is a solution such that exactly $\ell$ initial clusters have nonempty intersections with the composite clusters of the solution, then the probability that these $\ell$ clusters get distinct colors in a random coloring is  
at least $\frac{\ell!}{\ell^\ell}\geq e^{-\ell}\geq e^{-2B}$. Therefore, with probability at least $e^{-2B}$, a yes-instance admits a colorful solution.  

Our next task is to explain how to verify that there is a colorful solution for a given random coloring $\psi\colon\mathcal{J}\rightarrow \{1,\ldots,\ell\}$.

We need some auxiliary notation. For a set of colors $X\subseteq \{1,\ldots,\ell\}$, we use $\mathcal{J}(X)\subseteq \mathcal{J}$ to denote the subset of initial clusters with the colors from $X$ and $C(X)\subseteq\{1,\ldots,n\}$ is used to denote the set of indices in the initial clusters with their colors in $X$, that is, $C(X)=\cup_{J\in \mathcal{J}(X)}J$.  We also denote $\bfA(X)=\bfA[\{1,\ldots,m\},C(X)]$, that is, $\bfA(X)$ is the submatrix of $\bfA$ with the columns $\bfa_i$ such that $i\in C(X)$.    

Recall that we are looking for an acyclic solution $\mathcal{I}=\{I_1,\ldots,I_k\}$, that is, $G(\mathcal{I},\mathcal{J})$ is required to be  a forest. Let $\mathcal{I}$ be such a $k$-clustering. Let $\mathcal{I}'\subseteq \mathcal{I}$ be the set of composite clusters and let $\mathcal{J}'\subseteq \mathcal{J}$ be the set of initial clusters having nonempty intersections with the composite clusters. The main idea behind our algorithm for finding a colorful solution is to guess the structure of the forest $G(\mathcal{I}',\mathcal{J}')=G(\mathcal{I},\mathcal{J})[\mathcal{I}'\cup\mathcal{J}']$ and then do dynamic programming over it. Recall that $|\mathcal{I}'|=t$ and $|\mathcal{J}'|=\ell$ by our assumptions.  Note also that the leaves of $G(\mathcal{I}',\mathcal{J}')$ are initial clusters and every connected component of this forest contains at least three vertices.

We consider all forests $F$ on $t+\ell$ vertices such that  (i) each connected component of $F$ has at least three vertices, and
(ii) $F$ admits a bipartition $(U,W)$ of its vertex set with $|U|=t$ and $|W|=\ell$ such that  the leaves of $F$ are in $W$.  
Since $t\leq B$ and $\ell\leq 2B$, the number of such forests is $2^{\Oh(B)}$~\cite{Otter48} and they can be listed in $2^{\Oh(B)}$ time (see, e.g.,~\cite{WrightROM86}).  Note that since the leaves required to be in $W$, the bipartition $(U,W)$ is unique.
For a given forest $F$,  
we say that an acyclic $k$-clustering $\{I_1,\ldots,I_k\}$ for $\bfA$ is a \emph{feasible} (with respect to $F$ and the parameters $t$ and $\ell$) if the following holds:
\begin{itemize}
\item[(i)] $p\leq |I_i|\leq q$ for $i\in\{1,\ldots,k\}$,
\item[(ii)] the set $\mathcal{I}'\subseteq \mathcal{I}$ of composite clusters has size $t$ and the set $\mathcal{J}'\subseteq \mathcal{J}$ of initial clusters 
having nonempty intersections with the composite clusters has size $\ell$,
\item[(iii)] the initial clusters in $\mathcal{J}'$ are colored by distinct colors by $\psi$, and
\item[(iv)] $G(\mathcal{I}',\mathcal{J}')$ is isomorphic to $F$ with an isomorphism that bijectively maps $\mathcal{I}'$ to $U$ and $\mathcal{J}'$ to $W$.
\end{itemize}
The problem of finding a colorful solution boils down to checking whether there is $F$ such that there is a feasible $k$-clustering of cost at most $B$. We do the check by considering all the forests $F$. If we find that there is a feasible $k$-clustering  of cost at most $B$, we stop and return the yes-answer. Otherwise, we conclude that there is no colorful solution. 

Assume that a forest $F$ with the bipartition $(U,W)$ is given. Denote by $F_1,\ldots,F_f$ the connected components of $F$. 
Let $U_i=V(F_i)\cap U$ and $W_i=V(F_i)\cap W$ for $i\in\{1,\ldots,f\}$. Let also $t_i=|U_i|$ and $\ell_i=|W_i|$ for $i\in\{1,\ldots,f\}$.  

For $i\in\{1,\ldots,f\}$, $X\subseteq \{1,\ldots\ell\}$  and a positive integer $h\leq k$, denote by $\omega_i(X,h)$ the minimum cost of an $h$-clustering for $\bfA(X)$ 
that is feasible with respect to $F_i$ and the parameters $t_i$ and $\ell_i$ if $|X|=\ell_i$. 
We assume that $\omega_i(X,h)=+\infty$ if  $|X|\neq \ell_i$ or there is no an $h$-clustering that is feasible. 

We show that if we are given the tables of values of  $\omega_i(X,h)$, then we can verify whether  there a feasible $k$-clustering  of cost at most $B$.

\begin{lemma}\label{lem:comp-forest}
Given the values $\omega_i(X,h)$ for all $i\in\{1,\ldots,f\}$, $X\subseteq \{1,\ldots\ell\}$ and  positive integers $h\leq k$, 
it can decided in time $2^{\Oh(B)}\cdot n^2$ whether there is a feasible $k$-clustering for $\bfA$ of cost at most $B$ with respect to $F$, $t$ and $\ell$.
\end{lemma}

\begin{proof} To give the intuition behind the proof,
observe that a feasible $k$-clustering of cost at most $B$ with respect to $F$, $t$ and $\ell$ exists if and only if
there are positive integers $h_1,\ldots,h_f$ such that $h_1+\cdots+h_f=k$ and a partition $\{X_1,\ldots,X_f\}$ of $\{1,\dots,\ell\}$ such that 
\begin{equation*}\label{eq:comp-forest}
\omega_1(X_1,h_1)+\cdots+\omega_f(X_f,h_f)\leq B,
\end{equation*}
because in a feasible clustering the initial clusters in $\mathcal{J}'$ are colored by distinct colors. This leads to the following dynamic programming algorithm 

For $j\in\{1,\ldots,f\}$, $X\subseteq \{1,\ldots,\ell\}$, let $F^{(j)}$ be the disjoint union of $F_1,\ldots,F_j$, $t^{(j)}=t_1+\cdots+t_j$ and $\ell^{(j)}=\ell_1+\cdots+\ell_j$.
For $j\in\{1,\ldots,f\}$, $X\subseteq \{1,\ldots,\ell\}$ and positive integer $h$, denote by $w^{(j)}(X,h)$ the minimum cost of an $h$-clustering for $\bfA(X)$ that is feasible with respect to $F^{(j)}$, $t^{(j)}$ and $\ell^{(j)}$ if $|X|=\ell^{(j)}$; we also assume that  $w^{(j)}(X,h)=+\infty$ if $|X|\neq \ell^{(j)}$ or 
there is no feasible $h$-clustering.
Notice that $\omega_{1}(X,h)=w^{(1)}(X,h)$ and 
$w^{(f)}(X,h)$ is the minimum cost of an $h$-clustering for $\bfA(X)$ that is feasible with respect to $F$, $t$ and $\ell$.
Thus, $w^{(f)}(\{1,\ldots,\ell\},r)\leq B$ if and only if there is a feasible $k$-clustering for $\bfA$ of cost at most $B$ with respect to $F$, $t$ and $\ell$.

We compute the values of $w^{(j)}(X,h)$ for $j=1,2,\ldots,f$. As we observed,  $w^{(1)}(X,h)=\omega_1(X,h)$. To compute $w^{(j)}(X,h)$ for $j\geq 2$, we use the following recurrence:
\begin{equation}\label{eq:rec}
w^{(j)}(X,h)=\min\{\omega_j(Y,h')+w^{(j-1)}(X\setminus Y,h-h')\mid 1\leq h'<h\text{ and }\emptyset\neq Y\subset X \};
\end{equation}
we also assume that $w^{(j)}(X,h)=+\infty$ if the set in the right part of (\ref{eq:rec}) is empty. 

The correctness of (\ref{eq:rec}) is proved in the standard way by showing the two opposite inequalities. Let $X\subseteq \{1,\ldots,\ell\}$.
To simplify notation, assume that $\{J_1,\ldots,J_{s'}\}$ are initial clusters with colors from $X$. Let also $h\leq k$ be a positive integer. 

Suppose that $|X|=\ell^{(j)}$
and $\{I_1,\ldots,I_h\}$ is an $h$-clustering for $\bfA(X)$ that is feasible with respect to $F^{(j)}$, $t^{(j)}$ and $\ell^{(j)}$ of minimum cost.   
Let  $\mathcal{I}'\subseteq \{I_1,\ldots,I_h\}$ be the set of composite clusters and let $\mathcal{J}'\subseteq \{J_1,\ldots,J_{s'}\}$ be the set of initial clusters 
having nonempty intersections with the composite clusters. Recall that $|\mathcal{I}'|=t^{(j)}$, $|\mathcal{J}'|=\ell^{(j)}$, and the initial clusters in $\mathcal{J}'$ are colored by distinct colors. Consider an isomorphism $\alpha$ that bijectively maps the vertices of $G(\mathcal{I}',\mathcal{J}')$ to the vertices of $F$ with the property that 
the vertices of $\mathcal{I}'$ are mapped to $\bigcup_{i=1}^{(j)}U_i$ and $\mathcal{J}'$ are mapped  to $\bigcup_{i=1}^{(j)} W_i$.
Then $\ell_j$ clusters of $\mathcal{J}'$ are mapped to $W_j$. Denote  
by $Y\subset X$ the set of their colors. Clearly, $|Y|=\ell_j$ and $|X\setminus Y|=\ell^{(j)}-\ell_j=\ell^{j-1}$.
Notice that the clusters of $\mathcal{I}'$ that are mapped to $U_j$ are composed by elements of initial clusters with colors from $Y$ and no other composite cluster contains an element of an initial cluster with a color from $Y$.   
To simplify notation, assume that the clusters $I_1,\ldots,I_{h'}$ contain elements of the initial clusters with the colors from $Y$ and $I_{h'+1},\ldots,I_h$ are the clusters containing elements of the initial clusters with the colors from $X\subseteq Y$. Then we have that $\{I_1,\ldots,I_{h'}\}$ is a feasible $h'$-clustering for $\bfA(Y)$ with respect to $F_j$, $t_j$ and $\ell_j$. Similarly, we obtain that $\{I_{h'+1},\ldots,I_h\}$ is  a feasible $(h-h')$-clustering for $\bfA(X\setminus Y)$ with respect to $F^{(j-1)}$, $t^{(j-1)}$ and $\ell^{(j-1)}$.  Thus, $w^{(j)}(X,h)\geq\omega_j(Y,h')+w^{(j-1)}(X\setminus Y,h-h')$ and, therefore, 
\begin{equation}\label{eq:rec-one}
w^{(j)}(X,h)\geq \min\{\omega_j(Y,h')+w^{(j-1)}(X\setminus Y,h-h')\mid 1\leq h'<h\text{ and }\emptyset\neq Y\subset X \}.
\end{equation}
If either $|X|\neq \ell^{(j)}$ or there is  no an $h$-clustering for $\bfA(X)$ that is feasible with respect to $F^{(j)}$, $t^{(j)}$ and $\ell^{(j)}$, then $w^{(j)}(X,h)=+\infty$ and (\ref{eq:rec-one}) is trivial.

To show the opposite inequality, let nonempty $Y\subseteq X$ and positive $h'<h$ be such that the right part of (\ref{eq:rec}) is minimum. If 
$\omega_j(Y,h')+w^{(j-1)}(X\setminus Y,h-h')=+\infty$, then the required inequality holds trivially. Assume that 
this is not the case. 
Then $|Y|=\ell_j$, $|X\setminus Y|=\ell^{(j-1)}$,
there is an $h'$-clustering $\mathcal{I}^{(1)}$ for $\bfA(Y)$ of cost $\omega_j(Y,h')$ that is feasible with respect to $F_j$, $t_j$ and $\ell_j$, and there is an $(h-h')$-clustering $\mathcal{I}^{(2)}$ for $\bfA(X\setminus Y)$ of cost $w^{(j-1)}(X\setminus Y,h-h')$ that is feasible with respect to $F^{(j-1)}$, $t^{(j-1)}$ and $\ell^{(j-1)}$. Consider $\mathcal{I}=\mathcal{I}^{(1)}\cup\mathcal{I}^{(2)}$ and observe that this is an $h$-clustering for $\bfA(X)$ that is feasible with respect to $F^{(j)}$, $t^{(j)}$ and $\ell^{(j)}$. This means that 
$w^{(j)}(X,h)\leq \omega_j(Y,h')+w^{(j-1)}(X\setminus Y,h-h')$. By the choice of $Y$ and $h'$,
\begin{equation}\label{eq:rec-two}
w^{(j)}(X,h)\leq \min\{\omega_j(Y,h')+w^{(j-1)}(X\setminus Y,h-h')\mid 1\leq h'<h\text{ and }\emptyset\neq Y\subset X \}.
\end{equation}
 
 Combining (\ref{eq:rec-one}) and (\ref{eq:rec-two}), we obtain that the recurrence (\ref{eq:rec}) holds.
 
 Finally, we compute $w^{(f)}(X,h)$ for all $X\subseteq \{1,\ldots,\ell\}$ and all positive $h\leq r$. In particular, we find $w^{(f)}(\{1,\ldots,\ell\},k)$ and verify whether this value is at most $B$. 
 
 To evaluate the running time, note that to compute the table of values of $w^{(j)}(X,h)$ by (\ref{eq:rec}), we consider all nonempty $X$ of size at most $\ell$ and the nonempty subsets $Y\subset X$. This means that we consider at most $3^\ell$ pairs of sets. Also we consider all positive $h\leq k$ and $h'\leq h$, that is, at most $k^2$ pairs of integers. Since $\ell\leq 2B$ and $k\leq n$, the computations can be done it $2^{\Oh(B)}\cdot n^2$ time. Since $f\leq t \leq B$, the total running time is $2^{\Oh(B)}\cdot n^2$. 
\end{proof}

By Lemma~\ref{lem:comp-forest}, we have to compute the tables of values of $\omega_i(X,h)$ for all $i\in\{1,\ldots,f\}$, nonempty $X\subseteq\{1,\ldots,\ell\}$ and positive $h\leq k$. For this, we use the fact that $F_1,\ldots,F_f$ are trees and this allows us to use dynamic programming over these trees. 

\begin{lemma}\label{lem:tree}
Let $T$ be a tree with a bipartition $(U,W)$ of its vertex set such that  $t'=|U|\leq t$, $\ell'=|W|\leq \ell$ and 
 the leaves of $T$ are in $W$. For a given $X\subseteq \{1,\ldots,\ell\}$ with $|X|=\ell'$ and positive $h\leq k$, the minimum cost of a feasible $h$-clustering for $\bfA(X)$ with respect to $T$, $t'$ and $\ell'$ can be found in 
$2^{\Oh(B\log B)}|\Sigma|^B\cdot (mn)^{\Oh(1)}$ time.
\end{lemma}

\begin{proof}
We select a vertex $z\in U$ as a root of $T$. This selection defines a parent-child relation on the set of vertices. For a vertex $x\in V(T)$, we denote by $T_x$ the subtree of $T$ induced by the descendants of $x$ (including the vertex itself).  For $x\in V(T)$, let $t_x=V(T_x)\cap U$ and $\ell_x=V(T_x)\cap W$.
For every $x\in V(T)$, we compute the tables of auxiliary values depending on whether $x\in U$ or $x\in W$. 

For a set of colors $Z\subseteq X$, $J\in\mathcal{J}(Z)$ and $J'\subseteq J$, we use 
$\mathcal{J}(Z)/J'$ to denote the set of clusters obtained from the initial clusters of  $\mathcal{J}(Z)$ by the replacement of $J$ by $J''=J\setminus J'$ if $J'\subset J$
and $\mathcal{J}(Z)/ J'=\mathcal{J}(Z)\setminus \{J\}$ if $J'=J$.  We assume that the clusters of $\mathcal{J}(Z)/J'$ have the inherited colors. 
We also write $\bfA(Z)/J'$ to denote the submatrix of $\bfA(Z)$ obtained by the deletion of the columns with the indices from $J'$. Note that $\mathcal{J}(Z)/J'$ is the set of initial clusters for $\bfA(Z)/J'$.

\medskip
Suppose that $x\in W$. For every positive integer $h'\leq h$, every $Y\subseteq X$,  
every $c\in Y$, every $J\in \mathcal{J}(Y)$ and every nonnegative integer $j\leq |J|$, we define $\omega_x^{(1)}(h',Y,c,J,j)$. For technical reasons, it is convenient to define this function for leaves separately.

Let $x$ be a leaf. 
We define  $\omega_x^{(1)}(h',\{c\},c,J,j)$ as the minimum cost of an $h'$-clustering for $\bfA(Y)/J'$, where $J'\subseteq J$ of size $j$, such that all the clusters are simple, and $\omega_x^{(1)}(h',Y,c,J,j)=+\infty$ if $Y\neq\{c\}$.

If $x$ is an internal vertex of $T$, then $\omega_x^{(1)}(h',Y,c,J,j)$  
is the minimum cost of an $h'$-clustering $\mathcal{I}=\{I_1,
\ldots, I_{h'}\}$ for $\bfA(Y)/J'$, where $J'\subset J$ of size $j$, such that 
\begin{itemize}
\item[(i)] $p\leq |I_i|\leq q$ for $i\in\{1,\ldots,h'\}$,
\item[(ii)] the set $\mathcal{I}'\subseteq \mathcal{I}$ of composite clusters has size $t_x$, and the set $\mathcal{J}'\subseteq \mathcal{J}(Y)/J'$ of initial clusters 
having nonempty intersections with the composite clusters has size $\ell_x$, 
\item[(iii)] $|Y|=\ell_x$ and the initial clusters in $\mathcal{J}'$ are colored by distinct colors by $\psi$,
\item[(iv)] $G(\mathcal{I}',\mathcal{J}')$ is isomorphic to $T_x$ with an isomorphism $\alpha$ that bijectively maps $\mathcal{I}'$ to $U_x$, $\mathcal{J}'$ to $W_x$, and 
\item[(v)] $J\setminus J'\in \mathcal{J}'$, $\alpha(J\setminus J')=x$ and  $\psi(J\setminus J')=c$. 
\end{itemize}

In both cases, 
we assume that $\omega_x^{(1)}(h,Y,c,J,j)=+\infty$ if  there is no such an $h'$-clustering.

Informally, $\omega_x^{(1)}(h',Y,c,J,j)$ is the minimum cost of an $h'$-clustering for $\bfA(Y)/J'$ that is feasible with respect to $T_x$, $t_x$ and $\ell_x$ with the additional assumption that we take $j$ elements of $J$ colored by $c$ to include to the composite cluster that corresponds to the parent of $x$. 
Observe that the value of $\omega_x^{(1)}(h',Y,c,J,j)$ does not depend on the choice of $J'$.
Notice also that we have the special case when $U_x=\emptyset$, i.e., when $x$ is a leaf, because it this case we have no composite clusters. Then we form $h'$ simple clusters from the initial clusters $\mathcal{J}(Y)/J'$.  

\medskip
Let $x\in U$. For every positive integer $h'\leq h$, every $Y\subseteq X$, 
 every nonnegative integer $j\leq q$, and every $\bfs\in\mathcal{M}$, $\omega_x^{(2)}(h',Y,j,\bfs)$
is the minimum cost of an $h'$-clustering $\mathcal{I}=\{I_1,\ldots,I_{h'}\}$ for $\bfA(Y)$ such that
\begin{itemize}
\item[(i)] the cost of $I_1$ is computed with respect to the median $\bfs$, that is, the cost equals $\sum_{i\in I_1}\hdist(\bfs,\bfa_i)$,
\item[(ii)]  $p-j\leq |I_1|\leq q-j$ and   $p\leq |I_i|\leq q$ for $i\in\{2,\ldots,h'\}$,
\item[(iii)] for the set of composite clusters $\mathcal{I}'\subseteq \mathcal{I}$, $\mathcal{I}''=\mathcal{I}'\cup \{I_1\}$ has size $t_x$, and  the set $\mathcal{J}'\subseteq \mathcal{J}(Y)$ of initial clusters 
having nonempty intersections with the clusters from $\mathcal{I}''$ has size $\ell_x$, 
\item[(iv)] $|Y|=\ell_x$ and the initial clusters in $\mathcal{J}'$ are colored by distinct colors by $\psi$,
\item[(v)] $G(\mathcal{I}'',\mathcal{J}')$ is isomorphic to $T_x$ with an isomorphism $\alpha$ that bijectively maps $\mathcal{I}''$ to $U_x$, $\mathcal{J}'$ to $W_x$, and $\alpha(I_1)=x$.
\end{itemize}
In the same way as above for other functions, it is assumed that  $\omega_x^{(2)}(h',Y,j,\bfs)=+\infty$ if  there is no such an $h'$-clustering.

Informally, $\omega_x^{(2)}(h',Y,j,\bfs)$ is the minimum cost of an $h'$-clustering for $\bfA(Y)$ that is feasible with respect to $T_x$, $t_x$ and $\ell_x$, where the  specific cluster $I_1$ associated with $x$  is required to have $\bfs$ as its median and ``misses'' $j$ elements. Notice that it is not required that $\bfs$ is optimal for $I_1$. 
However, in future, $I_1$ is going to be complemented by $j$ elements of an initial cluster corresponding to the parent of $x$, unless $x$ is a root. 
Note also that $I_1$ is not a composite cluster if $x$ has a unique child, but because $I_1$ is expected to be complemented by other elements, 
 $I_1$ is counted as a composite cluster in the definition of  $\omega_x^{(2)}(h',Y,j,\bfs)$.

\medskip
Now we explain how to compute the table of values of $\omega_x^{(1)}(h',Y,c,J,j)$ and $\omega_x^{(2)}(h',Y,j,\bfs)$. First, we compute $\omega_x^{(1)}(h',Y,c,J,j)$ for leaves.

\begin{claim}\label{cl:leaves}
For every leaf $x$ of $T$, $\omega_x^{(1)}(h',Y,c,J,j)$ can be computed in $\Oh(n)$ time.
\end{claim}

\begin{proof}
If $Y\neq\{c\}$,  $\omega_x^{(1)}(h',Y,c,J,j)=+\infty$ by the definition. Assume that $Y=\{c\}$.
Let $J'\subseteq J$ be a set of size $j$. We compute $\hat{\mathcal{J}}=\mathcal{J}(Y)/J'$ in $\Oh(n)$ time. Then $\omega_x^{(1)}(h',Y,c,J,j)=0$ if every set in $\hat{\mathcal{J}}$ can be partitioned into clusters of size at least $p$ and at most $q$ in such a way that the total number of clusters is $h'$, and $\omega_x^{(1)}(h',Y,c,J,j)=+\infty$ otherwise. We apply Observation~\ref{obs:blocks}. First, we verify whether every $\hat{J}\in\hat{\mathcal{J}}$ can be partitioned into clusters of size at least $p$ and at most $q$ by checking whether $\Big\lceil\frac{|\hat{J}|}{q}\Big\rceil\leq \Big\lfloor\frac{|\hat{J}|}{p}\Big\rfloor$. If this holds, then we observe that we can obtain exactly $h'$ clusters in total if and only if 
 $\sum_{\hat{J}\in\hat{\mathcal{J}}}\Big\lceil\frac{|\hat{J}|}{q}\Big\rceil\leq h' \leq \sum_{\hat{J}\in\hat{\mathcal{J}}} \Big\lfloor\frac{|\hat{J}|}{p}\Big\rfloor$.
 Since checking of these conditions can be done in $\Oh(n)$ time, the total running time is $\Oh(n)$.
 \end{proof}

Next, we explain how to compute $\omega_x^{(1)}(h',Y,c,J,j)$ for internal vertices if the tables of values of $\omega_y^{(2)}(\cdot,\cdot,\cdot,\cdot)$ are given for all children $y$ of $x$. 

\begin{claim}\label{cl:internal-initial}
Let $x\in W$ be an internal vertex of $T$ and assume that the table of values of $\omega_y^{(2)}(\cdot,\cdot,\cdot,\cdot)$ is computed for every child $y$ of $x$. Then 
 $\omega_x^{(1)}(h',Y,c,J,j)$ can be computed in $2^{\Oh(B\log B)}|\Sigma|^B\cdot (mn)^{\Oh(1)}$ time.
 \end{claim}

\begin{proof}
Let  $h'\leq h$, $Y\subseteq X$ $c\in Y$, $J\in \mathcal{J}(Y)$ and  $j\leq |J|$.  
If $j=|J|$, then we immediately set $\omega_x^{(1)}(h',Y,c,J,j)=+\infty$, because we have no proper $J'\subset J$ of size $j$. 
Also if $|Y|\neq \ell_x$ or $\psi(J)\neq c$, then $\omega_x^{(1)}(h',Y,c,J,j)=+\infty$ by definition.
Assume that $j<|J|$, $J'\subset J$ of size $j$, $\psi(J)=c$ and  $|Y|=\ell_x$. Let $\hat{J}=J\setminus J'$.
We denote by $y_1,\ldots,y_f$ the children of $x$ in $T$. 

Consider the initial clusters of color $c$. By the definition of $\omega_x^{(1)}(h',Y,c,J,j)$, we are interested in an $h'$-clustering, where the initial clusters of color $c$ distinct from $J$ are split into simple clusters and, possibly, some parts of $J$ also form simple clusters. For a nonnegative integers $\hat{h}\leq h'$ and $\hat{j}\leq |\hat{J}|$, we define 
$w(\hat{h},\hat{j})$ to be $0$ if the initial clusters of  $\hat{\mathcal{J}}=\mathcal{J}(\{c\})/(J\setminus J'')$, where $J''\subseteq \hat{J}$ of size $\hat{j}$ can be partitioned into $\hat{h}$ simple clusters of of size at least $p$ and at most $q$, and we set  $w(\hat{h},\hat{j})=+\infty$ otherwise.
To compute $w(\hat{h},\hat{j})$, we use Observation~\ref{obs:blocks} similarly to the proof of Claim~\ref{cl:leaves}.  Namely, we verify whether every $\tilde{J}\in\hat{\mathcal{J}}$ can be partitioned into clusters of size at least $p$ and at most $q$ by checking whether $\Big\lceil\frac{|\tilde{J}|}{q}\Big\rceil\leq \Big\lfloor\frac{|\tilde{J}|}{p}\Big\rfloor$, and then we check whether  $\sum_{\tilde{J}\in\hat{\mathcal{J}}}\Big\lceil\frac{|\tilde{J}|}{q}\Big\rceil\leq \hat{h} \leq \sum_{\tilde{J}\in\hat{\mathcal{J}}} \Big\lfloor\frac{|\tilde{J}|}{p}\Big\rfloor$. Since $\hat{h}\leq h'\leq h$, the values of $w(\hat{h})$ can be computed in $\Oh(n^2)$ time. 

 Observe that by the definition of $\omega_x^{(1)}(h',Y,c,J,j)$, the elements of  $\hat{J}$ should be included in $f$ composite clusters associated with the children of $x$ in an $h'$-clustering for $\bfA(Y)/J'$. In particular, if $|\hat{J}|<f$, it cannot be done and $\omega_x^{(1)}(h',Y,c,J,j)=+\infty$ by the definition. From now, we assume that $|\hat{J}|\geq f$. 

For $i\in\{1,\ldots,f\}$, denote by $T^{(i)}$ the subtree of $T$ induced by  $\{x\}\cup\bigcup_{i'=1}^iV(T_{y_{i'}})$, set $U^{(i)}=U\cap V(T^{(i)})$  and $W^{(i)}=W\cap V(T^{(i)})$. Let also $t^{(i)}=|U^{(i)}|$ and $\ell^{(i)}=|W^{(i)}|$ for $i\in\{1,\ldots,f\}$. 
For each $i\in\{1,\ldots,f\}$, each nonnegative $\hat{h}\leq h'$, each positive $\hat{j}\leq |J|-j$, and every $c\in Z\subseteq Y$,  
denote by $w^{(i)}(\hat{h},\hat{j},Z)$, the minimum cost of $\hat{h}$-clustering $\mathcal{I}=\{I_1,\ldots,I_{\hat{h}}\}$ for $\bfA(Z)/(J\setminus J'')$, where $J''\subseteq \hat{J}$ of size $\hat{j}$, such that 
\begin{itemize}
\item[(i)] $p\leq |I_{i'}|\leq q$ for $i'\in\{1,\ldots,\hat{h}\}$,
\item[(ii)] the set $\mathcal{I}'\subseteq \mathcal{I}$ of composite clusters has size $t^{(i)}$, and the set $\mathcal{J}'\subseteq \mathcal{J}(Y)/(J\setminus J'')$ of initial clusters 
having nonempty intersections with the composite clusters has size $\ell^{(i)}$, 
\item[(iii)] $|Z|=\ell^{(i)}$ and the initial clusters in $\mathcal{J}'$ are colored by distinct colors by $\psi$,
\item[(iv)] $G(\mathcal{I}',\mathcal{J}')$ is isomorphic to $T^{(i)}$ with an isomorphism $\alpha$ that bijectively maps $\mathcal{I}'$ to $U^{(i)}$, $\mathcal{J}'$ to $W^{(i)}$, and 
\item[(v)] $J''\in \mathcal{J}'$, $\alpha(J'')=x$ and  $\psi(J'')=c$. 
\end{itemize}
We also follow the same convention as above that $w^{(i)}(\hat{h},\hat{j},Z)=+\infty$ if either there is no $\hat{h}$-clustering satisfying (i)--(v).
Observe that, by the definition, $\omega_x^{(1)}(h',Y,c,J,j)=w^{(f)}(h',|J|-j,Y)$. Therefore, we compute the tables of values of $w^{(i)}(\cdot,\cdot,\cdot)$ for $i=1,\ldots,f$.

To initiate the computation of $w^{(i)}(\cdot,\cdot,\cdot)$, it is convenient to formally define this function for $i=0$. We set
\begin{equation*}
w^{(0)}(\hat{h},\hat{j},Z)=
\begin{cases}
w(\hat{h},\hat{j})&\mbox{if }Z=\{c\},\\
+\infty&\mbox{otherwise}. 
\end{cases}
\end{equation*}
For $\bfs\in M$, denote $d(\bfs)=\hdist(\bfs,\bfa_j)$ for $j\in J$.
Then to compute $w^{(i)}(\hat{h},\hat{j},Z)$ for $i\geq 1$, we use the following recurrence:
\begin{equation}\label{eq:rec-init}
w^{(i)}(\hat{h},\hat{j},Z)=\min\{\omega^{(2)}_{y_i}(\hat{h}',\hat{Z},\hat{j}',\bfs)+\hat{j}'d(\bfs) +w^{(i-1)}(\hat{h}-\hat{h}',\hat{j}-\hat{j}',Z\setminus \hat{Z}) \},  
\end{equation}  
where the minimum in the right part is taken over all
integers $1\leq \hat{h}'\leq\hat{h}$ and $0<\hat{j}'\leq \hat{j}$, all sets $\hat{Z}$ such that $c\notin \hat{Z}\subset Z$, and all $\bfs\in\mathcal{M}$.
We assume that $w^{(i)}(\hat{h},\hat{j},Z)=+\infty$ if  the set in the right part is empty.
 
We prove the correctness of (\ref{eq:rec-init}) by showing the opposite inequalities between the left and the right part.

 If $w^{(i)}(\hat{h},\hat{j},Z)=+\infty$, then 
 \begin{equation*}
 w^{(i)}(\hat{h},\hat{j},Z)\geq \min\{\omega^{(2)}(\hat{h}',\hat{Z},\hat{j}',\bfs)+\hat{j}'d(\bfs)+w^{(i-1)}(\hat{h}-\hat{h}',\hat{j}-\hat{j}',Z\setminus \hat{Z})\}.
 \end{equation*} 
Suppose that  $w^{(i)}(\hat{h},\hat{j},Z)<+\infty$.  Consider  $\hat{h}$-clustering $\mathcal{I}$ for $\bfA(Z)/(J\setminus J'')$ of cost 
  $w^{(i)}(\hat{h},\hat{j},Z)$ satisfying (i)--(v).  Let  $I\in\mathcal{I}$ be the composite cluster such that $\alpha(I)=y_i$.  Since $\alpha(J'')=x$, $I$ contains elements of $J''$. Let $\hat{J}''=I\cap J''$ and 
  $\hat{j}'=\hat{J}''$. Denote by $\bfs\in \mathcal{M}$ the median of $I_1$.  Consider $\hat{\mathcal{J}}''=\alpha^{-1}(V(T_{y_i}))\cap\mathcal{J}'$, that is, the set of initial clusters having nonempty intersections with the composite clusters that are mapped by $\alpha$ to  the nodes of $T_{y_i}$. The coloring $\psi$ colors these clusters by distinct colors and we define $\hat{Z}$ to be the set of colors of the clusters of $\hat{\mathcal{J}}''$; note that $c\notin\hat{Z}$.  
Denote by $\hat{h}'$ the number of clusters in $\mathcal{I}$ containing elements of the initial clusters with colors in $\hat{Z}$ and let $\mathcal{I}_1$ be the set of these clusters; observe that $I\in\mathcal{I}_1$. Let $\mathcal{I}_2=\mathcal{I}\setminus \mathcal{I}_1$.  

By the defintion of the values of $w^{(i-1)}(\cdot,\cdot,\cdot)$, we obtain that the cost of clustering for $\mathcal{I}_2$ is at least $w^{(i-1)}(\hat{h}-\hat{h}',\hat{j}-\hat{j}',Z\setminus \hat{Z})$. The cluster $I$ contains $\hat{j}'$ elements of $J$. Since $\bfs$ is its median, these $\hat{j}'$ elements contribute $\hat{j}'d(\bfs)$ to its cost. Then, by the definition of $\omega^{(2)}_{y_i}(\cdot,\cdot,\cdot,\cdot)$, we have that the cost of clustering for $\mathcal{I}_1$ is at least
$\omega^{(2)}_{y_i}(\hat{h}',\hat{Z},\hat{j}',\bfs)+\hat{j}'d(\bfs)$. This means that 
 $w^{(i)}(\hat{h},\hat{j},Z)\geq \omega^{(2)}(\hat{h}',\hat{Z},\hat{j}',\bfs)+w^{(i-1)}(\hat{h}-\hat{h}',\hat{j}-\hat{j}',(Z\setminus \hat{Z})\cup\{c\})$ and
 \begin{equation}\label{eq:rec-init-one}
w^{(i)}(\hat{h},\hat{j},Z)\geq\min\{\omega^{(2)}_{y_i}(\hat{h}',\hat{Z},\hat{j}',\bfs)+\hat{j}'d(\bfs) +w^{(i-1)}(\hat{h}-\hat{h}',\hat{j}-\hat{j}',Z\setminus \hat{Z}) \}.  
\end{equation}

For the opposite direction, assume that integers $\hat{h}'$,  $\hat{j}'$, a set $\hat{Z}$, and a median $\bfs$ are chosen in such a way that the value of 
$\omega^{(2)}_{y_i}(\hat{h}',\hat{Z},\hat{j}',\bfs)+\hat{j}'d(\bfs) +w^{(i-1)}(\hat{h}-\hat{h}',\hat{j}-\hat{j}',Z\setminus \hat{Z})$ is minimum.
If the value is $+\infty$, then 
 $w^{(i)}(\hat{h},\hat{j},Z)\leq \omega^{(2)}(\hat{h}',\hat{Z},\hat{j}',\bfs)+\hat{j}'d(\bfs) +w^{(i-1)}(\hat{h}-\hat{h}',\hat{j}-\hat{j}',Z\setminus \hat{Z})$ as required.
 Assume that $\omega^{(2)}_{y_i}(\hat{h}',\hat{Z},\hat{j}',\bfs)<+\infty$ and  $w^{(i-1)}(\hat{h}-\hat{h}',\hat{j}-\hat{j}',Z\setminus \hat{Z})<+\infty$.  

By the definition of $\omega^{(2)}_{y_i}(\hat{h}',\hat{Z},\hat{j}',\bfs)$, there is an $\hat{h}'$-clustering $\mathcal{I}_1$ for $\bfA(\hat{Z})$  of cost $\omega^{(2)}_{y_i}(\hat{h}',\hat{Z},\hat{j}',\bfs)$ satisfying conditions (i)--(v) of the definition.  
 In particular, $\mathcal{I}_1$ contains a special cluster $I$ with the median $\bfs$ such that $p-\hat{j}'\leq |I|\leq q-\hat{j}'$ and $I$ is mapped to the root $y_i$ of $T_{y_i}$ by the isomorphism $\alpha$. 

Let $J''\subseteq \hat{J}$ of size $\hat{j}$ and let $\hat{J}''\subseteq J''$ of size $\hat{j}'$. By the definition of   $w^{(i-1)}(\hat{h}-\hat{h}',\hat{j}-\hat{j}',Z\setminus \hat{Z})$, there is an $(\hat{h}-\hat{h}')$-clustering $\mathcal{I}_2$ for $\bfA(Z\setminus \hat{Z})/((J\setminus J'')\cup\hat{J}'')$ 
of cost $w^{(i-1)}(\hat{h}-\hat{h}',\hat{j}-\hat{j}',Z\setminus \hat{Z})$ satisfying conditions (i)--(v) of the definition of  $w^{(i-1)}(\cdot,\cdot,\cdot)$.  

Observe that the clusters of $\mathcal{I}_1$ and $\mathcal{I}$ are pairwise disjoint and and include all elements of the initial clusters with their colors in $Z$ except $\hat{j}'+j$ elements of $J$.
We construct the $\hat{h}$-clustering $\mathcal{I}$ for $\bfA(Z)/(J\setminus J'')$ as follows. First, we modify the cluster $I\in\mathcal{I}_1$ by setting $I:=I\cup \hat{J}''$. Note that we increase the cost of the cluster by at most $\hat{j}'d(\bfs)$.
Then we take the union of $\mathcal{I}_1$ and $\mathcal{I}_2$. The definitions of the values  $\omega^{(2)}_{y_i}(\hat{h}',\hat{Z},\hat{j}',\bfs)$ and 
 $w^{(i-1)}(\hat{h}-\hat{h}',\hat{j}-\hat{j}',Z\setminus \hat{Z})$ imply that $\mathcal{I}$ satisfies conditions (i)--(v) for $w^{(i)}(\hat{h},\hat{j},Z)$. Therefore, 
 $w^{(i)}(\hat{h},\hat{j},Z)\leq \omega^{(2)}(\hat{h}',\hat{Z},\hat{j}',\bfs)+\hat{j}'d(\bfs) +w^{(i-1)}(\hat{h}-\hat{h}',\hat{j}-\hat{j}',Z\setminus \hat{Z})$.
 
 By the choice of   $\hat{h}'$,  $\hat{j}'$, $\hat{Z}$, and $\bfs$,  
\begin{equation}\label{eq:rec-init-two}
w^{(i)}(\hat{h},\hat{j},Z)\leq\min\{\omega^{(2)}_{y_i}(\hat{h}',\hat{Z},\hat{j}',\bfs)+\hat{j}'d(\bfs) +w^{(i-1)}(\hat{h}-\hat{h}',\hat{j}-\hat{j}',Z\setminus \hat{Z}) \}. \end{equation} 
 Then (\ref{eq:rec-init-one}) and (\ref{eq:rec-init-two}) imply (\ref{eq:rec-init}). 
 
 We use (\ref{eq:rec-init}) to compute the table of values of $w^{(f)}(\cdot,\cdot,\cdot)$.  Then $\omega_x^{(1)}(h',Y,c,J,j)=w^{(f)}(h',|J|-j,Y)$ by the definition. 
 
 To evaluate the running time, notice that the initial table $w^{(f)}(\cdot,\cdot,\cdot)$ can be computed in $2^{\Oh(B)}\cdot n^2$, since $w(\hat{h})$ can be computed in $\Oh(n^2)$ time and then the table is constructed for at most $n$ values of $\hat{j}$ and at most $2^\ell$ sets $Z$. To compute the table $w^{(i)}(\cdot,\cdot,\cdot)$ from $w^{(i-1)}(\cdot,\cdot,\cdot)$ by (\ref{eq:rec-init}) for $i\in\{1,\ldots,f\}$, we consider all pairs of integers $\hat{h}'\leq\hat{h}$, all pairs of sets $Z$ and $\hat{Z}\subset Z$ and all $\bfs\in\mathcal{M}$. Since $\hat{h}\leq n$, $Z\subseteq\{1,\ldots,\ell\}$ and $\ell\leq 2B$, and $|\mathcal{M}|=2^{\Oh(B\log B)}|\Sigma|^B\cdot (mn)^{\Oh(1)}$,  $w^{(i)}(\cdot,\cdot,\cdot)$ can be computed in $2^{\Oh(B\log B)}|\Sigma|^B\cdot (mn)^{\Oh(1)}$. Then the total running time is  $2^{\Oh(B\log B)}|\Sigma|^B\cdot (mn)^{\Oh(1)}$. 
\end{proof}

Further, we show how to compute $\omega_x^{(2)}(h',Y,j,\bfs)$ if the tables of values of $\omega_y^{(1)}(\cdot,\cdot,\cdot,\cdot,\cdot)$ are already computed.

\begin{claim}\label{cl:composite}
Let $x\in U$ be an internal vertex of $T$ and assume that the table of values of $\omega_y^{(1)}(\cdot,\cdot,\cdot,\cdot,\cdot)$ is computed for every child $y$ of $x$. Then 
 $\omega_x^{(2)}(h',Y,j,\bfs)$ can be computed in $2^{\Oh(B)}\cdot n^{\Oh(1)}$ time.
 \end{claim}

\begin{proof}
Let $h'\leq h$, $Y\subseteq X$, $j\leq q$, and let $\bfs\in \mathcal{M}$. If $|Y|\neq \ell_x$, then  $\omega_x^{(2)}(h',Y,j,\bfs)=+\infty$ by defintion. Assume that
$|Y|=\ell_x$. In the same way as in the proof of Claim~\ref{cl:internal-initial}, denote by $y_1,\ldots,y_f$ the children of $x$ in $T$. 
For $i\in\{1,\ldots,f\}$, let $T^{(i)}$ be the subtree of $T$ induced by  $\{x\}\cup\bigcup_{i'=1}^iV(T_{y_{i'}})$, set $U^{(i)}=U\cap V(T^{(i)})$  and $W^{(i)}=W\cap V(T^{(i)})$. Let also $t^{(i)}=|U^{(i)}|$ and $\ell^{(i)}=|W^{(i)}|$ for $i\in\{1,\ldots,f\}$. For an initial cluster $J$, we denote by $d(J)=\hdist(\bfs,\bfa_i)$ for $i\in J$. 
Similarly to the proof of Claim~\ref{cl:internal-initial}, we compute some auxiliary values. 

For each $i\in\{1,\ldots,f\}$, every positive integer  $\hat{h}\leq h'$, every nonnegative integer $\hat{j}\leq q$, and every nonempty $Z\subseteq X$, $w^{(i)}(\hat{h},\hat{j},Z)$ is 
 the minimum cost of an $\hat{h}$-clustering $\mathcal{I}=\{I_1,\ldots,I_{\hat{h}}\}$ for $\bfA(Z)$ such that
\begin{itemize}
\item[(i)] the cost of $I_1$ is computed with respect to the median $\bfs$, that is, the cost equals $\sum_{i\in I_1}\hdist(\bfs,\bfa_i)$,
\item[(ii)]  $|I_1|=\hat{j}$ and  $p\leq |I_{i'}|\leq q$ for $i'\in\{2,\ldots,\hat{h}\}$,
\item[(iii)] for the set of composite clusters $\mathcal{I}'\subseteq \mathcal{I}$, $\mathcal{I}''=\mathcal{I}'\cup \{I_1\}$ has size $t^{(i)}$, and  the set $\mathcal{J}'\subseteq \mathcal{J}(Z)$ of initial clusters 
having nonempty intersections with the clusters from $\mathcal{I}''$ has size $\ell^{(i)}$, 
\item[(iv)] $|Z|=\ell^{(i)}$ and the initial clusters in $\mathcal{J}'$ are colored by distinct colors by $\psi$,
\item[(v)] $G(\mathcal{I}'',\mathcal{J}')$ is isomorphic to $T^{(i)}$ with an isomorphism $\alpha$ that bijectively maps $\mathcal{I}''$ to $U^{(i)}$, $\mathcal{J}'$ to $W^{(i)}$, and $\alpha(I_1)=x$.
\end{itemize}
We assume that   $w^{(i)}(\hat{h},\hat{j},Z)=+\infty$ if  there is no such a $\hat{h}$-clustering.

Notice that the parameter $\hat{j}$ defines the size of a selected cluster $I_1$. Then, by the definition, we have that 
\begin{equation}\label{eq:omega-two}
\omega_x^{(2)}(h',Y,j,\bfs)=\min\{w^{(f)}(h',\hat{j},Y)\mid p-j\leq \hat{j}\leq q-j\}
\end{equation}
assuming that $\omega_x^{(2)}(h',Y,j,\bfs)=+\infty$ if the set in the right part is empty.

We compute the tables of values of $w^{(i)}(\cdot,\cdot,\cdot)$ for $i=1,\ldots,f$.

First, we observe that 
\begin{equation}\label{eq:omega-two-first}
w^{(1)}(\hat{h},\hat{j},Z)=\min\{\omega_{y_1}^{(1)}(\hat{h}-1,Z,c,J,\hat{j})+\hat{j}d(J)\mid c\in Z,J\in\mathcal{J}(Z)\};
\end{equation}
as before, $w^{(1)}(\hat{h},\hat{j},Z)=+\infty$ if the set in the right part is empty.

To see that $w^{(1)}(\hat{h},\hat{j},Z)\geq\min\{\omega_{y_1}^{(1)}(\hat{h}-1,Z,c,J,\hat{j})+\hat{j}d(J)\mid c\in Z,J\in\mathcal{J}(Z)\}$, assume that 
 $w^{(1)}(\hat{h},\hat{j},Z)<+\infty$; otherwise, the inequality is trivial. Let $\mathcal{I}=\{I_1,\ldots,I_{\hat{h}}\}$ be an $\hat{h}$-clustering for $\bfA(Z)$ satisfying conditions (i)--(v). Since $y_1$ is the unique child of $x$ in $T^{(1)}$, $I_1$ consists of  $\hat{j}$ elements of some initial cluster $J$. Let $c$ be the color assigned to $J$ by $\psi$.  Then, by the definition of $\omega_{y_1}^{(1)}(\hat{h}-1,Z,c,J,\hat{j})$, $\{I_2,\ldots,I_{\hat{h}}\}$ is an $(\hat{h}-1)$-clustering for $\bfA(Z)/J'$ for $J'\subseteq J$ of size $\hat{j}$ that satisfies all the condition of the definition of  $\omega_{y_1}^{(1)}(\cdot,\cdot,\cdot,\cdots,\cdot)$. Therefore, the cost of $\{I_2,\ldots,I_{\hat{h}}\}$ is an $(\hat{h}-1)$ is at least $\omega_{y_1}^{(1)}(\hat{h}-1,Z,c,J,\hat{j})$. The median of $I_1$ is $\bfs$ and $I_1$ contains $\hat{j}$ elements of $J$. Therefore, the cost of $I_1$ is $\hat{j}d(J)$. We conclude that 
$w^{(1)}(\hat{h},\hat{j},Z)\geq\omega_{y_1}^{(1)}(\hat{h}-1,Z,c,J,\hat{j})+\hat{j}d(J)$. Therefore,  
  $w^{(1)}(\hat{h},\hat{j},Z)\geq\min\{\omega_{y_1}^{(1)}(\hat{h}-1,Z,c,J,\hat{j})+\hat{j}d(J)\mid c\in Z,J\in\mathcal{J}(Z)\}$. 
 
 Now we prove that $w^{(1)}(\hat{h},\hat{j},Z)\leq\min\{\omega_{y_1}^{(1)}(\hat{h}-1,Z,c,J,\hat{j})+\hat{j}d(J)\mid c\in Z,J\in\mathcal{J}(Z)\}$. If the right part of (\ref{eq:omega-two-first}) is $+\infty$, then the inequality is trivial. Assume that this is not the case and let  $c\in Z$ and $J\in\mathcal{J}(Z)$ be such that the   right part of (\ref{eq:omega-two-first}) achieves the minimum value for them. Then there is an $(\hat{h}-1)$-clustering $\mathcal{I}$ for $\bfA(Z)/J'$, where $J'\subseteq J$ has size $\hat{j}$, with the cost $\omega_{y_1}^{(1)}(\hat{h}-1,Z,c,J,\hat{j})$ that  satisfies all the condition of the definition of  $\omega_{y_1}^{(1)}(\cdot,\cdot,\cdot,\cdots,\cdot)$.  Then we construct a new cluster $I=J'$ with the median $\bfs$. Clearly, the cost is $\hat{j}d(J)$. It is straightforward to verify that $\mathcal{I}\cup\{I\}$ satisfies (i)--(v). Therefore, 
  $w^{(1)}(\hat{h},\hat{j},Z)\leq\omega_{y_1}^{(1)}(\hat{h}-1,Z,c,J,\hat{j})+\hat{j}d(J)$ and 
  $w^{(1)}(\hat{h},\hat{j},Z)\leq\min\{\omega_{y_1}^{(1)}(\hat{h}-1,Z,c,J,\hat{j})+\hat{j}d(J)\mid c\in Z,J\in\mathcal{J}(Z)\}$. 
 
 Combining the two inequalities, we conclude that (\ref{eq:omega-two-first}) holds.
 
 To compute $w^{(1)}(\hat{h},\hat{j},Z)$ for $i\geq 2$, we show that
 \begin{equation}\label{eq:omega-two-step}
w^{(i)}(\hat{h},\hat{j},Z)=\min\{\omega_{y_i}^{(1)}(\hat{h}',\hat{Z},c,J,\hat{j}')+\hat{j}'d(J)+w^{(i-1)}(\hat{h}-\hat{h}',\hat{j}-\hat{j}',Z\setminus \hat{Z})\},
\end{equation}
where the minimum is taken over all positive integers $\hat{h}'<\hat{h}$, $\hat{j}'<\hat{j}$, all nonempty sets $\hat{Z}\subset Z$, all $c\in Z$, and $J\in \mathcal{J}(Z)$. As it is standard in our paper,  $w^{(i)}(\hat{h},\hat{j},Z)=+\infty$ if the set in the right part of (\ref{eq:omega-two-step}) is empty.  
 
 We prove (\ref{eq:omega-two-step}) by demonstrating the opposite inequalities between the left and the right part.
 
 If $w^{(i)}(\hat{h},\hat{j},Z)=+\infty$, then 
 $w^{(i)}(\hat{h},\hat{j},Z)\geq\min\{\omega_{y_i}^{(1)}(\hat{h}',\hat{Z},c,J,\hat{j}')+\hat{j}'d(J)+w^{(i-1)}(\hat{h}-\hat{h}',\hat{j}-\hat{j}',Z\setminus \hat{Z})\}$.
 Assume that this is not the case. Then there is an $\hat{h}$-clustering  $\mathcal{I}$ for $\bfA(Z)$ of cost
 $w^{(i)}(\hat{h},\hat{j},Z)$ satisfying (i)--(v). In particular, there is $I\in\mathcal{I}$ such that $|I|=\hat{j}$, $\alpha(I)=x$ and $\bfs$ is its median.
 Let $J\in\mathcal{J}(Z)$ be the initial cluster such that $\alpha(J)=y_i$.  Denote by $c$ its color.
 By definition, $J\cap I\neq\emptyset$. Let $J'=I\cap J$ and $\hat{j}'=|J'|$.
 Consider $\hat{\mathcal{J}}'=\alpha^{-1}(V(T_{y_i}))\cap \mathcal{J}'$, that is, the set of initial clusters intersecting composite clusters that are mapped by $\alpha$ to the vertices of $T_{y_i}$. Note that $J\in \hat{\mathcal{J}}'$.
 By definition, these clusters are colored by distinct colors by $\psi$. Denote by $\hat{Z}$ the set of their colors. Clearly, $c\in\hat{Z}$.
 Let $\mathcal{I}_1\subseteq \mathcal{I}\setminus \{I\}$ be the set of clusters in $\mathcal{I}$ having nonempty intersections with with the initial clusters from $\hat{\mathcal{J}}'$; note that $I\notin \mathcal{I}_1$ by definition.  Set $\hat{h}'=|\mathcal{I}_1|$. Let $\mathcal{I}_2=\mathcal{I}\setminus \mathcal{I}_1$.

Observe that $\mathcal{I}_1$ is an $\hat{h}'$-clustering for $\bfA(\hat{Z})/J'$. Moreover, $\mathcal{I}_1$ satisfies all the conditions of the definition of 
$\omega_{y_i}^{(1)}(\hat{h}',\hat{Z},c,J,\hat{j}')$. This implies that the cost of $\mathcal{I}_1$ is at least $\omega_{y_i}^{(1)}(\hat{h}',\hat{Z},c,J,\hat{j}')$.
Consider the clustering $\hat{\mathcal{I}}_2$ obtained from $\mathcal{I}_2$ by the replacement of $I$ by $\hat{I}=I\setminus J'$. Notice that 
the clusters of 
$\hat{\mathcal{I}}_2$ contains only elements of initial clusters with colors from $Z\setminus \hat{Z}$.  Also we have hat $|\hat{I}|=\hat{j}-\hat{j}'$ and $|\hat{\mathcal{I}}_2|=\hat{h}-\hat{h}'$, because $i\geq 2$ and $I\neq J'$. Then it is straightforward to verify that $\hat{\mathcal{I}}_2$ is $(\hat{h}-\hat{h}')$-clustering for $\bfA(Z\setminus \hat{Z})$ satisfying (i)--(v) for $w^{(i-1)}(\cdot,\cdot,\cdot)$.
Therefore, the cost of $\hat{\mathcal{I}}_2$ is at least $w^{(i-1)}(\hat{h}-\hat{h}',\hat{j}-\hat{j}',Z\setminus \hat{Z})$. Finally, recall that $J'\subset I$. Since $\bfs$ is the median of $I$, the contribution of $J'$ to the cost is $\hat{j}'d(J)$. We conclude that
 $w^{(i)}(\hat{h},\hat{j},Z)\geq\omega_{y_i}^{(1)}(\hat{h}',\hat{Z},c,J,\hat{j}')+\hat{j}'d(J)+w^{(i-1)}(\hat{h}-\hat{h}',\hat{j}-\hat{j}',Z\setminus \hat{Z})$. Hence,
 \begin{equation}\label{eq:omega-two-upper}
  w^{(i)}(\hat{h},\hat{j},Z)\geq\min\{\omega_{y_i}^{(1)}(\hat{h}',\hat{Z},c,J,\hat{j}')+\hat{j}'d(J)+w^{(i-1)}(\hat{h}-\hat{h}',\hat{j}-\hat{j}',Z\setminus \hat{Z})\}. 
\end{equation} 
  
The opposite inequality is trivial if the right part of (\ref{eq:omega-two-step}) equals $+\infty$. Assume that this is not the case and suppose that  positive integers $\hat{h}'<\hat{h}$, $\hat{j}'<\hat{j}$, a set $\hat{Z}\subset Z$, $c\in Z$, and $J\in \mathcal{J}(Z)$ are chosen in such a way that the right part of (\ref{eq:omega-two-step}) achieves the minimum value for them. 

By the definition of $\omega_{y_i}^{(1)}(\hat{h}',\hat{Z},c,J,\hat{j}')$, there is an $\hat{h}'$-clustering $\mathcal{I}_1$ for $\bfA(\hat{Z})/J'$ of cost  $\omega_{y_i}^{(1)}(\hat{h}',\hat{Z},c,J,\hat{j}')$ satisfying conditions (i)--(v) of the definition, where $J'\subseteq J$ of size $\hat{j}'=|J'|$.  In particular, $c$ is a color of $J$. 
   
We also have that, by definition of  $w^{(i-1)}(\hat{h}-\hat{h}',\hat{j}-\hat{j}',Z\setminus \hat{Z})$, there is an $(\hat{h}-\hat{h}')$-clustering for $\bfA(Z\setminus \hat{Z})$ satisfying conditions (i)-(v) of the definition. In particular, there is a special cluster $I\in\mathcal{I}_2$ of size $\hat{j}-\hat{j}'$ with the median $\bfs$. 

We construct the clustering $\mathcal{I}$ for $\bfA(Z)$ as follows. First, we modify the cluster $I\in\mathcal{I}_2$ by replacing it by $I'=I\cup J'$. Then we take the union of $\mathcal{I}_1$ and the modified $\mathcal{I}_2$. It is straightforward to verify that $\mathcal{I}$ is a $\hat{h}$-clustering for $\bfA(Z)$ satisfying (i)--(v) for $w^{(i)}(\hat{h},\hat{j},Z)$. Since $I'$ is obtained by adding $\hat{j}'$ elements of $J$, the cost of $\mathcal{I}$ is 
$\omega_{y_i}^{(1)}(\hat{h}',\hat{Z},c,J,\hat{j}')+\hat{j}'d(J)+w^{(i-1)}(\hat{h}-\hat{h}',\hat{j}-\hat{j}',Z\setminus \hat{Z})$. Therefore,
$w^{(i)}(\hat{h},\hat{j},Z)\leq\omega_{y_i}^{(1)}(\hat{h}',\hat{Z},c,J,\hat{j}')+\hat{j}'d(J)+w^{(i-1)}(\hat{h}-\hat{h}',\hat{j}-\hat{j}',Z\setminus \hat{Z})$ and, by the choice of $\hat{h}'$, $\hat{j}'$, $\hat{Z}$, $c$ and $J$,
 \begin{equation}\label{eq:omega-two-lower}
  w^{(i)}(\hat{h},\hat{j},Z)\leq\min\{\omega_{y_i}^{(1)}(\hat{h}',\hat{Z},c,J,\hat{j}')+\hat{j}'d(J)+w^{(i-1)}(\hat{h}-\hat{h}',\hat{j}-\hat{j}',Z\setminus \hat{Z})\}. 
\end{equation} 
  
By (\ref{eq:omega-two-upper}) and (\ref{eq:omega-two-lower}), we conclude that the recurrence (\ref{eq:omega-two-step}) holds. Then we compute the tables of values of $w^{(i)}(\cdot,\cdot,\cdot)$ for $i=1,\ldots,f$ using (\ref{eq:omega-two-first}) and (\ref{eq:omega-two-step}). Finally, we apply (\ref{eq:omega-two})
to compute $\omega_x^{(2)}(h',Y,j,\bfs)$.

Clearly, the table of values of $w^{(1)}(\cdot,\cdot,\cdot)$  can be computed in $2^{\Oh(B)}\cdot n^3$ time, because we consider $\hat{h},\hat{j}\leq n$ and at most $2^\ell$ sets $Z$, and then go through at most $\ell$ values of $c$ and at most $n$ sets $\mathcal{J}$.  To compute the tables of values of  $w^{(i)}(\cdot,\cdot,\cdot)$ for $i\geq 2$, we consider all pairs of integers $\hat{h}'<\hat{h}$, all pairs $\hat{j}'<\hat{j}$, all nonempty sets $\hat{Z}\subset Z$, all $c\in Z$, and $J\in \mathcal{J}(Z)$.  Since $\hat{h}',\hat{h},\hat{j}',\hat{j}\leq n$, the number of pairs of set $\hat{Z}\subset Z$ is at most $3^\ell$, the number of the choices of $c$ is at most $\ell$ and the number of the choices of $J$ is at most $n$, we have that the total running time is $2^{\Oh(B)}\cdot n^{\Oh(1)}$, because $\ell\leq 2B$.
\end{proof}

Claims~\ref{cl:leaves}--\ref{cl:composite} allow us to compute the table of values of $\omega_z^{(2)}(\cdot,\cdot,\cdot,\cdot)$ for the root $z$ of $T$ bottom-up starting from the leaves (recall that $z\in U$). To make the final step of our algorithm, observe that   
 the minimum cost of a feasible $h$-clustering for $\bfA(X)$ with respect to $T$, $t'$ and $\ell'$ is
 \begin{equation*}
 \min\{\omega_x^{(2)}(h,X,0,\bfs)\mid \bfs\in\mathcal{M}\}
 \end{equation*}
 by the definition of these values.

The tree $T$ has $\ell'+t'\leq 3B$ vertices. The table of values of either $\omega_x^{(1)}(\cdot,\cdot,\cdot,\cdot,\cdot)$ or $\omega_x^{(2)}(\cdot,\cdot,\cdot,\cdot)$ 
 constructed for every node $x$ has size $2^{\Oh(B\log B)}|\Sigma|^B\cdot (mn)^{\Oh(1)}$ and can be constructed in 
 $2^{\Oh(B\log B)}|\Sigma|^B\cdot (mn)^{\Oh(1)}$ time by Claims~\ref{cl:leaves}--\ref{cl:composite}. Therefore, the total running time is 
 $2^{\Oh(B\log B)}|\Sigma|^B\cdot (mn)^{\Oh(1)}$.
 \end{proof}
 
 Using Lemmas~\ref{lem:comp-forest} and \ref{lem:tree}, we are able to check whether the considered instance has a colorful solution.
 
 \begin{lemma}\label{lem:colorful} 
 Given positive integers $t\leq B$ and  $\ell$ such that $t+1\leq \ell\leq 2B$ and a coloring $\psi\colon \mathcal{J}\rightarrow \{1,\ldots,\ell\}$, it can be decided in 
 $2^{\Oh(B\log B)}|\Sigma|^B\cdot (mn)^{\Oh(1)}$ time whether $(\bfA,\Sigma,k,B,p,q)$ has an acyclic colorful solution with $t$ composite clusters such that exactly $\ell$ initial clusters have nonempty intersections with the composite clusters of the solution. 
 \end{lemma}

\begin{proof}
Recall that we consider all forests $F$ on $t+\ell$ vertices such that  (i) each connected component of $F$ has at least three vertices, and
(ii) $F$ admits a bipartition $(U,W)$ of its vertex set with $|U|=t$ and $|W|=\ell$ such that  the leaves of $F$ are in $W$.  
We list $2^{\Oh(B)}$ forests (see~\cite{Otter48}) in  $2^{\Oh(B)}$ time (see~\cite{WrightROM86}). Then for each $F$, we check whether there is 
a $k$-clustering $\mathcal{I}=\{I_1,\ldots,I_k\}$ for $\bfA$ of cost at most $B$ that is a feasible with respect to $F$,  $t$ and $\ell$. For this, we consider  connected components $F_1,\ldots,F_f$ of $F$ and use Lemma~\ref{lem:tree} to compute the tables of values of $\omega_i(X,h)$ for all $i\in\{1,\ldots,f\}$, nonempty $X\subseteq\{1,\ldots,\ell\}$ and positive $h\leq r$. This can be done in $2^{\Oh(B\log B)}|\Sigma|^B\cdot (mn)^{\Oh(1)}$ time.
Then we apply the algorithm from Lemma~\ref{lem:comp-forest} that in $2^{\Oh(B)}\cdot n^2$ time checks whether the required $\mathcal{I}$ exists.
If we find that there is $\mathcal{I}$ of cost at most $B$ for the considered $F$, we stop and return the yes-answer. Otherwise, we obtain that there is no solution.
To complete the proof, it remains no note that the total running time is $2^{\Oh(B\log B)}|\Sigma|^B\cdot (mn)^{\Oh(1)}$.
\end{proof}
 
 Recall that if there is a solution with $t$ composite clusters such that exactly $\ell$ initial clusters have nonempty intersections with the composite clusters of the solution, then the probability that these $\ell$ clusters are assigned distinct colors in a random coloring $\psi$ is at least $e^{-2k}$. Then the probability that some initial clusters having nonempty intersections with the composite clusters of the solution obtain the same color is at most $1-e^{-2B}$. This implies that if we try $e^{2B}$ random colorings, then the probability that for every coloring, some initial clusters having nonempty intersections with the composite clusters of the solution are of the same color is at most $(1-e^{-2B})^{e^{2B}}\leq e^{-1}$.  This leads to the following randomized algorithm. We consider $N=\lceil e^{2B}\rceil$ random coloring $\psi$, and for each coloring, we verify the existence of a colorful solution using Lemma~\ref{lem:colorful}. If a colorful solution exists for $\psi$, then we report that  $(\bfA,\Sigma,r,k,p,q)$ admits a required solution and stop. 
 Otherwise, if we fail to find a colorful solution for every $\psi$, we conclude that with the probability at least $e^{-1}$ there is no solution. 
 The running time of this algorithm is $e^{2B}\cdot 2^{\Oh(B\log B)}|\Sigma|^B\cdot (mn)^{\Oh(1)}$, that is, $2^{\Oh(B\log B)}|\Sigma|^B\cdot (mn)^{\Oh(1)}$.
 
 This algorithm can be derandomized by standard tools~\cite{AlonYZ95} (see also~\cite[Chapter~5]{CyganFKLMPPS15}).  More precisely, we replace random colorings by functions   
  from a perfect hash family.
 
Let $s$ and $\ell$ be positive integers such that $s\geq \ell$. A set $\mathcal{F}$ of functions $\xi \colon\{1,\ldots,s\}\rightarrow\{1,\ldots,\ell\}$ is said to be an \emph{$(s,\ell)$-perfect hash family} if for every $X\subseteq \{1,\ldots,s\}$ of size $\ell$, there is $\xi\in \mathcal{F}$ such that $\xi|_X$ is a bijection between $X$ and $\{1,\ldots,\ell\}$.
 
We use the result of Naor, Schulman and Srinivasan~\cite{NaorSS95} (see also~\cite[Chapter~5]{CyganFKLMPPS15}).

\begin{proposition}\label{prop:hash}
For every $s\geq \ell\geq 1$, there is an $(s,\ell)$-perfect hash family $\mathcal{F}$ of size $e^{\ell}\ell^{\Oh(\log \ell)}\cdot \log s$ that can be constructed in
 $e^{\ell}\ell^{\Oh(\log\ell)}\cdot s\log s$ time.
 \end{proposition}
 
 We consider our set of initial clusters $\mathcal{J}=\{J_1,\ldots,J_s\}$ and construct an $(s,\ell)$-perfect hash family $\mathcal{F}$. Since $\ell\leq 2B$ and $s\leq n$, $|\mathcal{F}|=e^{2B}(2B)^{\Oh(\log B)}\cdot \log n$ and $\mathcal{F}$ can be constructed in  $e^{2B}(2B)^{\Oh(\log B)}\cdot n\log n$ time by Proposition~\ref{prop:hash}.  For every $\xi\in\mathcal{F}$, we define the coloring $\psi_\xi\colon\mathcal{J}\rightarrow\{1,\ldots,\ell\}$ by setting $\psi_\xi(J_i)=\xi(i)$ for $i\in\{1,\ldots,s\}$. 
 
 If $(\bfA,\Sigma,k,B,p,q)$ admits a solution with $t$ composite clusters such that exactly $\ell$ initial clusters have nonempty intersections with the composite clusters, then there is $\xi\in\mathcal{F}$ such that $\psi_\xi$ colors these initial clusters by distinct colors by the definition of an $(s,\ell)$-perfect hash family. Then our randomized algorithm can be restated as follows. For each $\xi\in F$,  we verify the existence of a colorful solution with respect to $\psi_\xi$ using Lemma~\ref{lem:colorful}. If a colorful solution exists for some $\psi_\xi$, then we report that  $(\bfA,\Sigma,k,B,p,q)$ admits a solution and stop. Otherwise, we conclude that there is no such a solution. This immediately gives the following lemma.
 
 \begin{lemma}\label{lem:fixed-ell} 
 Given positive integers $t\leq B$ and $\ell$ such that $t+1\leq \ell\leq 2B$, it can be decided in 
 $2^{\Oh(B\log B)}|\Sigma|^B\cdot (mn)^{\Oh(1)}$ time whether $(\bfA,\Sigma,k,B,p,q)$ has an acyclic solution with $t$ composite clusters such that exactly $\ell$ initial clusters have nonempty intersections with the composite clusters of the solution. 
 \end{lemma}
 
 Recall that we try all the possible values of the number of composite clusters $t\leq B$.
 Since for any acyclic solution, the number 
 $\ell$ of initial clusters having nonempty intersections with the composite clusters is at least $t+1$ and at most $2B$, by trying all $\ell$ within this interval, we can decide in  $2^{\Oh(B\log B)}|\Sigma|^B\cdot (mn)^{\Oh(1)}$ time whether $(\bfA,\Sigma,k,B,p,q)$ has an acyclic solution with $t$ composite clusters. Recall that by Lemma~\ref{lem:forest},   $(\bfA,\Sigma,k,B,p,q)$ is a yes-instance if and only if it admits an acyclic solution. We obtain that we can solve \probCClust for $(\bfA,\Sigma,k,B,p,q)$ in   $2^{\Oh(B\log B)}|\Sigma|^B\cdot (mn)^{\Oh(1)}$ time and this concludes the proof of Theorem~\ref{thm:FPT-k}.
 
\section{Clustering with size constraints}\label{sec:variants} 
 In this section, we discuss   other variants of \probClust with cluster size constraints: 
  \probBClust and  \probFBClust. We also discuss 
 the special case of \probCClust for $p=q=n/k$ that is equivalent   to \probBClust for $\delta=0$  and to \probFBClust for $\alpha=1$. We refer to this problem as \probEClust.
 
 Recall that by Theorem~\ref{thm:NPhard},  \probCClust is \classNP-complete for $k=2$ and $p=q=n/2$, that is, \probEClust  is \classNP-complete for $k=2$. Using the same arguments as in the proof of Theorem~\ref{thm:NPhard}, we can show the following more general claim.

 \begin{theorem}\label{thm:NPhard-balance}
For every  fixed $\alpha\geq 1$ ($\delta\geq 0$, respectively), \probFBClust (\probBClust, respectively) is \classNP-complete for $k=2$ and binary matrices.
\end{theorem} 
 
 From the positive side, we observe that  \probBClust and \probFBClust admit Turing reductions to \probCClust, that is, \probCClust is the most general among the considered problems. For this, we make the following straightforward observation.

\begin{observation}\label{obs:reduction}
An instance $(\bfA,\Sigma,k,B,\delta)$ of \probBClust (an instance $(\bfA,\Sigma,k,B,\alpha)$ of \probFBClust, respectively) is a yes-instance if and only if there is a nonnegative integer $p$ such that  $\frac{n}{k}-\delta\leq p\leq \frac{n}{k}$ ( $\frac{n}{\alpha k}\leq p\leq \frac{n}{k}$, respectively) and for $q=p+\delta$ ($q=\alpha p$, respectively), $(\bfA,\Sigma,k,B,p,q)$ is a yes-instance of \probCClust. 
\end{observation}

Thus, given an algorithm $\mathcal{A}$ for \probCClust, we can solve  \probBClust for $(\bfA,\Sigma,k,B,\delta)$ as follows. We consider all $p$ starting from 
$\max\{1,\lceil\frac{n}{k}\rceil-\delta\}$ up to $\lfloor\frac{n}{k}\rfloor$, and use $\mathcal{A}$ to solve \probCClust for $(\bfA,\Sigma,k,B,p,\min\{n,p+\delta\})$. If $\mathcal{A}$ returns ``yes'' for one of the values of $p$, we conclude that $(\bfA,\Sigma,k,B,\delta)$ is a yes-instance of \probBClust and stop. Otherwise, if $\mathcal{A}$ always returns ``no'', $(\bfA,\Sigma,k,B,\delta)$ is a no-instance. 
Clearly, \probFBClust can be solved in similar way.  This allows to obtain the following corollary of Theorem~\ref{thm:FPT-k}.

 \begin{corollary}\label{cor:FPT-k}
\probBClust and \probFBClust are solvable in time $2^{\Oh(B\log B)}|\Sigma|^B\cdot (mn)^{\Oh(1)}$. 
\end{corollary}   

\section{Conclusion}\label{sec:concl}
 We proved that \probCClust can be solved in $2^{\Oh(B\log B)}|\Sigma|^B\cdot (mn)^{\Oh(1)}$ time. This also implies that the same holds for \probBClust and \probFBClust. The natural question is whether it is possible to improve the dependence on $B$? We do not know the answer to this question even for the special case of \probEClust.
 
 Another important direction of research in the investigation of kernelization for clustering problems with size constraints. In~\cite[Theorem~3]{FominGP20}, Fomin, Golovach and Panolan proved that \probClust  
 does not admit a polynomial kernel when parmeterized by $B$, unless $\classNP\subseteq\classCoNP/{\sf poly}$. This immediately implies the following proposition.
   
 \begin{proposition}\label{thm:no-kern}   
\probCClust (\probBClust and \probFBClust, respectively) has no polynomial kernel when paramterized by $B$, unless $\classNP\subseteq\classCoNP/{\sf poly}$, even if $\Sigma=\{0,1\}$. 
\end{proposition}    
 
 Also by Theorems~\ref{thm:NPhard} and~\ref{thm:NPhard-balance} the problems are already \classNP-hard for $k$. Thus, for kernelization, we have to consider more restrictive parameterizations. 
 Up to now, we have only partial results. In particular, we can show  \probBClust admits a polynomial kernel when parameterzied by $B$, $k$ and $\delta$. 
 
We start with some auxiliary results. First, we observe that if there is an initial cluster $J$ of size at least $B+1$, then at least one median should be the same as a column of the input matrix with the index in $J$.
 
\begin{observation}\label{obs:init-mean}
 Let $\{I_1,\ldots,I_k\}$ be an $k$-clustering for a matrix $\bfA=(\bfa_1,\ldots,\bfa_n)$ of cost at most $B$ and let $J\subseteq\{1,\ldots,n\}$ be an initial cluster with $|J|\geq B+1$.  Then there is $i\in\{1,\ldots,k\}$ such that an optimal median of $I_i$ coincides with $\bfs=\bfa_j$ for $j\in J$.
\end{observation}   
 
\begin{proof}
For the sake of contradiction, assume that  medians $\bfc_1,\ldots,\bfc_k$ for the clusters $\{I_1,\ldots,I_k\}$, respectively, are distinct from $\bfs$. Then
\begin{equation*}
 \sum_{i=1}^{k}\sum_{j\in I_i}\hdist(\bfc_i,\bfa_j) \geq \sum_{i=1}^{k} \sum_{j \in J \cap I_i}\hdist(\bfc_i,\bfs)\geq |J|>B
\end{equation*}
contradicting that the cost of $\{I_1,\ldots,I_k\}$ is at most $B$.
\end{proof}   
   
Our next lemma shows that if there is a clustering such that a median $\bfc_i$ coincides with a column $\bfa_j$, then  we can either collect all the elements of the initial cluster $J$ containing $j$ in the same cluster of a solution or form a cluster of a solution out of its elements.

   \begin{lemma}\label{lem:exchange}
 Let $\{I_1,\ldots,I_k\}$ be a $k$-clustering for a matrix $\bfA=(\bfa_1,\ldots,\bfa_n)$ with optimal medians $\bfc_1,\ldots,\bfc_k$, respectively.
 Let also $\bfS\subseteq\{\bfc_1,\ldots,\bfc_k\}$ be the set of medians coinciding  with columns of $\bfA$.
 Then there is a $k$-clustering $\{I_1',\ldots,I_k'\}$ for $\bfA$ such that  
 \begin{itemize}
 \item[(i)] $|I_i'|=|I_i|$ for all $i\in\{1,\ldots,k\}$, 
 \item[(ii)] $\sum_{i=1}^k\sum_{j\in I_i'}\hdist(\bfc_i,\bfa_j)\leq \sum_{i=1}^k\sum_{j\in I_i}\hdist(\bfc_i,\bfa_j)$, and 
 \item[(iii)] for every $\bfs\in\bfS$ and the initial cluster $J$ such that $\bfs=\bfa_j$ for $j\in J$, there is $i\in\{1,\ldots,k\}$ such that 
 either $J\subseteq I_i'$ or $I_i'\subset J$.
 \end{itemize}
 \end{lemma}

 \begin{proof}
 Let $\bfc_1,\ldots,\bfc_k$ be optimal medians for $I_1,\ldots,I_k$, respectively. Assume without loss of generality that $\bfS=\{\bfc_1,\ldots,\bfc_t\}$, and denote by $J_1,\ldots,J_t$ the initial clusters such that for every $i\in\{1,\ldots,t\}$, $\bfa_j=\bfc_i$ for $j\in J_i$. 
 Let $\mathcal{I}'=\{I_1',\ldots,I_k'\}$ be a $k$-clustering for $\bfA$ such that (a)~$|I_i'|=|I_i|$ for all $i\in\{1,\ldots,k\}$, (b) $\sum_{i=1}^k\sum_{j\in I_i'}\hdist(\bfc_i,\bfa_j)\leq \sum_{i=1}^k\sum_{j\in I_i}\hdist(\bfc_i,\bfa_j)$, and (c) $\sum_{i=1}^t|I_i'\cap J_i|$ is maximum. 
 We claim that $\mathcal{I}'$ satisfies conditions (i)--(iii) of the lemma. Clearly, (i) and (ii) are fulfilled by conditions (a) and (b) of the choice of $\mathcal{I}'$.  To show (iii), we prove that either $J_i\subseteq I_i'$ or $I_i'\subset J_i$ for every $i\in\{1,\ldots,t\}$.
 
 Assume to the contrary that there is $i\in\{1,\ldots,t\}$ such that 
 neither $J_i\subseteq I'_i$ nor $I'_i \subset J_i$.
 Then there is a cluster $I'_j$ for $j \in \{1,\ldots,k\}$ such that $j\neq i$, $I'_j\cap J \neq \emptyset$, and there is $h \in I'_i$ such that $h \notin J_i$. Let $\ell \in I'_j\cap J_i$.
 Consider the $k$-clustering $\mathcal{I}''=\{I_1'',\ldots,I_k''\}$ such that 
 $I_i''=(I_i'\cup\{\ell\})\setminus \{h\}$, $I_j''=(I_j'\cup\{h\})\setminus\{\ell\}$, and $I_h''=I_h'$ for $h\in\{1,\ldots,k\}$ such that $h\neq i,j$. In words, we exchange the elements $h$ and $\ell$ between $I_i'$ and $I_j'$. 
 Then
 \begin{equation*}
 \sum_{p=1}^k\sum_{q\in I_p'}\hdist(\bfc_p,\bfa_q)- \sum_{p=1}^k\sum_{q\in I_p''}\hdist(\bfc_p,\bfa_q)=\hdist(\bfc_i,\bfa_h)+\hdist(\bfc_j,\bfa_\ell)-\hdist(\bfc_i,\bfa_\ell)-\hdist(\bfc_j,\bfa_h), 
 \end{equation*}
 and  since $\bfc_i=\bfa_\ell$, we obtain that
  \begin{equation*}
 \sum_{p=1}^k\sum_{q\in I_p'}\hdist(\bfc_p,\bfa_q)- \sum_{p=1}^k\sum_{q\in I_p''}\hdist(\bfc_p,\bfa_q)=\hdist(\bfa_\ell,\bfa_h)+\hdist(\bfc_j,\bfa_\ell)-\hdist(\bfc_j,\bfa_h)\geq 0 
 \end{equation*}  
 by the triangle inequality.
 This means that    
 \begin{equation}\label{eq:b-sat}
 \sum_{p=1}^k\sum_{q\in I_p''}\hdist(\bfc_p,\bfa_q)\leq \sum_{p=1}^k\sum_{q\in I_p'}\hdist(\bfc_p,\bfa_q)\leq \sum_{p=1}^k\sum_{q\in I_p}\hdist(\bfc_p,\bfa_q).
 \end{equation}
 Since $|I_i''|=|I_i'|$ for all $i\in\{1,\ldots,r\}$, $\mathcal{I}''$ satisfies condition (a) of the choice of $\mathcal{I}'$. Condition (b) is satisfied because of (\ref{eq:b-sat}). However, 
 $|I_i''\cap J|=|(I_1'\cap I)\cup\{\ell\}|=|I_i'\cap J|+1$. Because $\mathcal{I}''$ was obtained by the exchange $h$ and $\ell$ between $I_i'$ and $I_j'$,
 $I_p'\cap J_p\subseteq I_p''\cap J_p$ for $p\in\{1,\ldots,t\}$. We obtain 
 that $\sum_{p=1}^t|I_p'\cap J_p|<\sum_{p=1}^t|I_p''\cap J_p|$   
  contradicting (c). Therefore, either $J_p\subseteq I_p'$ or $I_p'\subset J_p$ for every $p\in\{1,\ldots,t\}$
 as it claimed.
  \end{proof}  

The following lemma is used to find medians if the sizes of clusters in a solution are sufficiently big.

\begin{lemma}\label{lem:gen-means}
Let $\mathcal{I}=\{I_1,\ldots,I_k\}$ be a $k$-clustering for a matrix $\bfA=(\bfa_1,\ldots,\bfa_n)$ of cost at most $B$ such that 
$s\leq |I_i|\leq s+\delta$ for all $i\in\{1,\ldots,k\}$, where $\delta$ is a nonnegative integer and an integer $s\geq 2B+1+(k-1)\delta$. Then for every initial clusters $J\subseteq \{1,\ldots,n\}$, the following is fulfilled for $\bfc=\bfa_j$ for $j\in J$:
\begin{itemize}
\item[(i)] if $|J|\mod s\geq B+1+(k-1)\delta$, then exactly $\Big\lceil\frac{|J|}{s}\Big\rceil$ clusters of $\mathcal{I}$ have optimal medians coinciding with $\bfc$ (the other medians are different),
\item[(ii)] if $|J|\mod s\leq B+(k-1)\delta$, then exactly $\Big\lfloor\frac{|J|}{s}\Big\rfloor$ clusters of $\mathcal{I}$ have optimal medians coinciding with $\bfc$.
\end{itemize}
\end{lemma}

\begin{proof}
We start with proving (i). Let  $|J|\mod s\geq B+1+(k-1)\delta$. We show that (i) holds for $J$ by induction on $p=\Big\lfloor\frac{|J|}{s}\Big\rfloor$.

The base case is $p=0$. Then $\ceil*{\frac{|J|}{s}}=1$.  As $|J| \mod s \geq B+1+(k-1)\delta$ and $\floor*{\frac{|J|}{s}}=0$, $B+1 \leq |J| \leq s$.
By Observation~\ref{obs:init-mean}, there is a cluster in $\mathcal{I}$ whose optimal median is $\bfc$. Thus, at least one optimal median coincides with $\bfc$.
Without loss of generality, we assume that $\bfc$ is the median of $I_1$.
We now show that $\bfc_i\neq\bfc$ for $i\in\{2,\ldots,k\}$.
Assume  to the contrary that there exists $h\in\{2,\ldots,k\}$ such that $\bfc_h=\bfc$.
By Lemma~\ref{lem:exchange}, there is a $k$-clustering $\mathcal{I}=\{I'_1,\ldots,I'_k\}$  for $\bfA$ such that $|I_i'|=|I_i'|$ for all $i\in\{1,\ldots,k\}$, $\sum_{i=1}^k\sum_{j\in I_i'}\hdist(\bfc_i,\bfa_j)\leq \sum_{i=1}^r\sum_{j\in I_i}\hdist(\bfc_i,\bfa_j)$, and $J\subseteq I_1'$. 
Then
\begin{equation*}
\sum_{i=1}^k\sum_{j\in I_i'}\hdist(\bfc_i,\bfa_j)\geq \sum_{j\in I_h'}\hdist(\bfc_h,\bfa_j)=\sum_{j\in I_h'}\hdist(\bfc,\bfa_j)\geq |I_h'|\geq s\geq B+1,
\end{equation*}
contradicting that $\cost(\mathcal{I})\leq B$. We conclude that exactly one median coincides with $\bfc$, that is, (i) holds for $p=0$.

Now let $p \geq 1$ and assume that the claim holds when $p$ is smaller. Note that $k\geq 2$ in this case.
We observe that, because $|J|\mod s\geq B+1+(k-1)\delta$,  $|J|\geq sp+B+1+(k-1)\delta$. 
 By  Observation~\ref{obs:init-mean}, there is a cluster in $\mathcal{I}$ whose optimal median is $\bfc$. 
Without loss of generality, we assume that $\bfc$ is the median of $I_1$.
Then 
 by Lemma~\ref{lem:exchange}, there is  a $k$-clustering $\{I_1',\ldots,I_k'\}$ for $\bfA$ such that  
 $|I_i'|=|I_i|$ for all $i\in\{1,\ldots,k\}$,  $\sum_{i=1}^k\sum_{j\in I_i'}\hdist(\bfc_i,\bfa_j)\leq \sum_{i=1}^k\sum_{j\in I_i}\hdist(\bfc_i,\bfa_j)$, and 
$I_1'\subset J$. 
Consider $\bfA'=\bfA[\{1,\ldots,m\},\{1,\ldots,n\}\setminus I_1]$, that is, $\bfA'$ is obtained from $\bfA$ by the deletion of the columns with their indices in $I_1$.
Notice that $\mathcal{I}'=\{I_2',\ldots,I_k'\}$ is an $(k-1)$-clustering for $\bfA'$ of cost at most $B$. Moreover, because $|I_i'|=|I_i|\geq s\geq 2B+1$, $\bfc_2,\ldots,\bfc_k$ are unique optimal medians for $I_2',\ldots,I_k'$, respectively, by Observation~\ref{obs:big-clust}. 
Let $J'=J\setminus I_1'$. 
Since $|I_1'|\leq s+\delta$, 
\begin{equation*}
|J'|=|J|-|I'_1| \geq sp+B+1+(k-1)\delta-s-\delta=s(p-1)+B+1+(k-2)\delta\geq B+1+(k-2)\delta.
\end{equation*}
By our inductive hypothesis, exactly $\ceil*{{\frac{|J'|}{s}}}$ clusters of $\mathcal{I}'$ have optimal medians coinciding with $\bfc$. 
As $|I_1|\geq s$, $\floor*{{\frac{|J'|}{s}}}\leq p-1$. Because  $|J'|\geq s(p-1)+B+1+(k-2)\delta$, $\floor*{{\frac{|J'|}{s}}}\geq p-1$. Hence, $\floor*{{\frac{|J'|}{s}}}=p-1$ and $\ceil*{{\frac{|J'|}{s}}}=p$.
Since $\bfc_2,\ldots,\bfc_k$ are optimal medians, exactly $p$  of them are equal to $\bfc$. 
Together with the median $\bfc_1=\bfc$, exactly $p+1$ medians in $\{\bfc_1,\ldots,\bfc_k\}$ are equal to $\bfc$. 
Then  exactly $\Big\lceil\frac{|J|}{s}\Big\rceil=p+1$ clusters of $\mathcal{I}$ have optimal medians coinciding with $\bfc$.
This completes the proof of (i). 

\medskip
To show (ii), 
we first claim that for every initial cluster $J$, there are at least $p=\floor*{\frac{|J|}{s}}$ clusters in $\mathcal{I}$, whose optimal medians are equal to $\bfc$, where $\bfc=\bfa_j$ for $j\in J$.  
The proof is by induction on $p$.

The claim is trivial if $p=0$. Let $p\geq 1$ and assume that the claim holds when $p$ is smaller. Since $p\geq 1$, 
$|J|\geq s\geq B+1$. By Observation~\ref{obs:init-mean}, there is a cluster in $\mathcal{I}$ whose optimal median is $\bfc$. 
Without loss of generality, we assume that $\bfc$ is the median of $I_1$.
Then by Lemma~\ref{lem:exchange}, there is  an $r$-clustering $\{I_1',\ldots,I_k'\}$ for $\bfA$ such that  
 $|I_i'|=|I_i|$ for all $i\in\{1,\ldots,k\}$,  $\sum_{i=1}^k\sum_{j\in I_i'}\hdist(\bfc_i,\bfa_j)\leq \sum_{i=1}^k\sum_{j\in I_i}\hdist(\bfc_i,\bfa_j)$, and 
either $J\subseteq I_1'$ or  $I_1'\subset J$. 

Suppose that $J \subseteq I_1'$. Then $|J|\leq |I_i'|\leq s+\delta<2s$. This means that $p=1$ and our claim holds, as $\bfc=\bfc_1$.

Assume from now that this is not the case, that is, $I_i'\subset J$. Then we argue similarly to the proof of (i).
Consider $\bfA'=\bfA[\{1,\ldots,m\},\{1,\ldots,n\}\setminus I_1]$, that is, $\bfA'$ is obtained from $\bfA$ by the deletion of the columns with their indices in $I_1$.
Notice that $\mathcal{I}'=\{I_2',\ldots,I_k'\}$ is an $(r-1)$-clustering for $\bfA'$ of cost at most $B$. Moreover, because $|I_i'|=|I_i|\geq s\geq 2B+1$, $\bfc_2,\ldots,\bfc_k$ are unique optimal medians for $I_2',\ldots,I_k'$, respectively, by Observation~\ref{obs:big-clust}. 
Let $J'=J\setminus I_1'$. 

If $\floor*{{\frac{|J'|}{s}}}\geq p-1$, then by the inductive assumption, there are at least $p-1$ clusters in $\mathcal{I}'$, whose optimal medians  coincide with $\bfc$. Thus, at least $p-1$ medians from $\{\bfc_2,\ldots,\bfc_k\}$ are equal to $\bfc$. Taking into account $\bfc_1=\bfc$, we have that at least $p$ medians from $\{\bfc_1,\ldots,\bfc_k\}$ are equal to $\bfc$, as required.

Let $\floor*{{\frac{|J'|}{s}}}\leq p-2$. Note that $p\geq 2$ in this case. Since $|I_1'|\leq s+\delta$ and $|J|\geq ps$, we obtain 
that $|J'|=|J|-|I_1'|\geq (p-2)s+(s-\delta)$. Thus, $\floor*{{\frac{|J'|}{s}}}=p-2$ and
\begin{equation*}
|J'|\mod s\geq s-\delta\geq 2B+1+(k-1)\delta\geq B+1+(k-2)\delta.
\end{equation*} 
By the already proven (i), we have that there are at least $\ceil*{{\frac{|J'|}{s}}}=p-1$ clusters in $\mathcal{I}'$, whose optimal medians  coincide with $\bfc$. 
Since $\bfc_1=\bfc$, we again obtain that at least $p$ medians from $\{\bfc_1,\ldots,\bfc_k\}$ are equal to $\bfc$. This concludes the proof of our auxiliary claim.

To finish the proof of (ii), assume that  $|J|\mod s\leq B+(k-1)\delta$. We already have that at least $p=\Big\lfloor\frac{|J|}{s}\Big\rfloor$ clusters of $\mathcal{I}$ have optimal medians coinciding with $\bfc$. It remains to show that there are at most $p$ such clusters. 
Assume to the contrary that  at least $p+1$ medians are equal to $\bfs$ and assume without loss of generality that $\bfc=\bfc_1=\ldots=\bfc_{p+1}$. 
Then
\begin{align*}
\sum_{i=1}^{k}\sum_{j \in I_i}\hdist(\bfc_i,\bfa_j) \geq&\sum_{i=1}^{p+1}\sum_{j \in I_i}\hdist(\bfc_i,\bfa_jj)=\sum_{j \in I_1\cup\ldots \cup I_{p+1}}\hdist(\bfc,\bfa_j)\\\geq&\sum_{j \in (I_1\cup\ldots \cup I_{p+1})\setminus J}\hdist(\bfc,\bfa_j)\geq |(I_1\cup\ldots\cup I_{p-1})\setminus I|.
\end{align*}
We know that  $ |I_i| \geq s$ for $i \in \{1,\ldots,k\}$. Then $|I_1\cup\ldots\cup I_{p+1}| \geq s(p+1)$. Since  $|J| \mod s \leq B+(k-1)\delta$, $|J|\leqslant ps+ B+(k-1)\delta$. 
This implies $|(I_1\cup\ldots\cup I_{p+1}) \setminus J|\geq s-B-(k-1)\delta  \geqslant B+1$. Hence $\sum_{j \in (I_1\cup\ldots \cup I_{P+1}) \setminus J}\hdist(\bfc,\bfa_j) \geq B+1>B$ contradicting that $\cost(\mathcal{I})\leq B$. This proves that exactly $\Big\lfloor\frac{|J|}{s}\Big\rfloor$ clusters of $\mathcal{I}$ have optimal medians coinciding with $\bfc$.
\end{proof}

Lemma~\ref{lem:gen-means} allows us to compute optimal medians and solve \probBClust if the average size of clusters is sufficiently big.

\begin{lemma}\label{lem:polynomialtimealgorithm}
\probBClust can be solved in polynomial time for instances $(\bfA,\Sigma,k,B,\delta)$ with $\frac{n}{k}\geq 2B+1+\delta k$. 
\end{lemma}

\begin{proof}
Let $(\bfA,\Sigma,k,B,\delta)$ be an instance of \probBClust with $\frac{n}{k}\geq 2B+1+\delta k$. Clearly, we can assume that $\delta\leq n-1$.
If $(\bfA,\Sigma,k,B,\delta)$ is a yes-instance, then there is an integer $s$ such that $\frac{n}{k}-\delta \leq s \leq \frac{n}{k}$ and 
$s\leqslant |I_i| \leqslant s+\delta$ for a solution  $\{I_1,\ldots,I_k\}$ to the instance. 

Then we consider all integers $s$ such that $\frac{n}{k}-\delta \leq s \leq \frac{n}{k}$. For each value of $s$, we check whether there is 
a solution $\{I_1,\ldots,I_k\}$ for the considered instance
with $s\leq |I_i| \leq s+\delta$, for all $i \in \{1,\ldots,k\}$.  If yes, we  return the yes-answer, otherwise, if we fail to find a solution for every $s$, then the algorithm returns the no-answer. 

Let $s$ be fixed. 
For each initial cluster $J$, we compute $\floor*{\frac{|J|}{s}}$ and $|J|\mod s$. Using these two values, we find the medians coinciding with $\bfc$ such that $\bfc=\bfa_j$ for $j\in J$ using Lemma~\ref{lem:gen-means}.  Denote by $\mathcal{C}$ the obtained collection of medians. If $|\mathcal{C}|\neq k$, then we discard the current choice of $s$. Otherwise, $\mathcal{C}$ contains exactly $k$ potential medians and we combine Observation~'\ref{obs:reduction} and Lemma~\ref{lem:means-clusters} to decide whether 
$(\bfA,\Sigma,k,B,\delta)$ admits a solution with these medians.

Since we consider at most $\delta+1\leq n$ values of $s$ and the algorithm from Lemma~\ref{lem:means-clusters} is polynomial, the total running time of our algorithm is polynomial. 
\end{proof}

In~\cite{FominGP20}, Fomin et al. proved that  \probClust admits a polynomial kernel when  parameterized by $B$ and $k$ for the binary matrices.
As one of the steps of their kernelization algorithm (see~Theorem~2 of~\cite{FominGP20}), they show that the number of rows in the output matrix can be reduced to $\Oh(B(B+r))$. Formally, the proof is done for the binary case, that is, for $\Sigma=\{0,1\}$, but it does not depend on $\Sigma$. 
We state this result in the following lemma.

\begin{lemma}[\cite{FominGP20}]\label{lem:fedoretallresult}
There is a polynomial algorithm that, given an instance $(\bfA,\Sigma,k,B)$ of \probClust with $m \times n$ matrix $\bfA$, produces an equivalent instance $(\bfA',\Sigma,k,B)$ with $m' \times n$ matrix $\bfA'$ such that the following holds:
\begin{itemize}
\item $m'=\Oh(B(B+k))$.
\item $\{I_1,\ldots,I_k\}$ is a solution for $(\bfA,\Sigma,k,B)$ if and only if it is also a solution for  $(\bfA',\Sigma,k,B)$.
\end{itemize} 
\end{lemma}

Now we are ready to show a polynomial kernel for \probBClust.

\begin{theorem}\label{thm:kernel}
\probBClust admits a kernel, where the output matrix has 
$\Oh(B(B+k))$ rows and  $O(k(B+\delta k))$ columns, and is a matrix over an alphabet of size at most $B+k$.
\end{theorem}

\begin{proof}
Let $(\bfA,\Sigma,k,B,\delta)$ be an instance of \probBClust with $\bfA=(\bfa_1,\ldots,\bfa_n)$.

Suppose $\frac{n}{k}\geq 2B+1+\delta k$. Then, by Lemma~\ref{lem:polynomialtimealgorithm},  the problem can be solved in polynomial time. 
We do it and return a trivial yes or no-instance, respectively. 
For example, we can return either the matrix $(0,0)$ or $(0,1)$, respectively, and set $k=1$, $B=0$ and $\delta=0$.
Assume from now that  $\frac{n}{k}\leq 2B+\delta k$, that is, $n \leq 2Bk+\delta k^2$. 

If $\bfA$ has at least $B+k+1$ pairwise distinct columns, then for every $k$-clustering $\{I_1,\ldots,I_k\}$ and 
every $\bfc_1,\ldots,\bfc_k\in\Sigma^m$,
$\sum_{i=1}^k\sum_{j\in I_i}\hdist(\bfc_i,\bfa_j)\geq B+1$, because at least $B+1$ columns of $\bfA$ are distinct from each median.
Thus, $(\bfA,\Sigma,k,B,\delta)$ is a no-instance in this case, and we return a trivial no-instance of \probBClust.

Assume from now that the number of pairwise distinct columns is at most $B+k$. If $|\Sigma|>B+k$, 
then we can replace every symbol of $\Sigma$ by a symbol of 
$\Sigma'=\{0,\ldots,B+k-1\}$ maintaining the following property: for each row of $\bfA$, the same symbols of $\Sigma$ are replaced by the same symbols of $\Sigma'$. It is straightforward to verify that this replacement produces an equivalent instance, because we are using the Hamming distances. From now, we assume that $|\Sigma|\leq B+k$. 

Given $(\bfA,\Sigma,k,B,\delta)$, we consider the instance  $(\bfA,\Sigma,k,B)$ of \probClust. We use algorithm from Lemma~\ref{lem:fedoretallresult} and 
denote by $(\bfA',\Sigma,k,B)$ the output instance. Then we construct the instance $(\bfA',\Sigma,k,B,\delta)$ of \probBClust and output it. 
 
 We show that  $(\bfA,\Sigma,k,B,\delta)$ is a yes-instance of \probBClust if and only if  $(\bfA',\Sigma,k,B,\delta)$ is a yes-instance. 
 
For the forward direction, suppose $(\bfA,\Sigma,k,B,\delta)$ is a yes-instance of \probBClust. Let $\mathcal{I}=\{I_1,\ldots,I_k\}$ be a solution to the instance. Clearly, $\mathcal{I}$ is a solution for the instance  $(\bfA,\Sigma,k,B)$ of \probClust. By Lemma~\ref{lem:fedoretallresult}, $\mathcal{I}$ is a solution for $(\bfA',\Sigma,k,B)$. Then $\mathcal{I}$ is a solution for $(\bfA',\Sigma,k,B,\delta)$. For the opposite direction, the arguments are similar.
Let $\mathcal{I}=\{I_1,\ldots,I_k\}$  be a solution for  $(\bfA',\Sigma,k,B,\delta)$. Then this is a solution for the instance $(\bfA',\Sigma,k,B)$ of \probClust and, by Lemma~\ref{lem:fedoretallresult}, a solution for $(\bfA',\Sigma,k,B)$. Finally, $\mathcal{I}$ is a solution of  $(\bfA,\Sigma,k,B,\delta)$.  

Recall that $n=\Oh(k(B+\delta k))$ and note that $\bfA'$ has $\Oh(B(B+k))$  rows by Lemma~\ref{lem:fedoretallresult}. Since $|\Sigma|\leq B+k$, we conclude 
that the output matrix has 
$\Oh(B(B+r))$ rows and  $O(k(B+\delta k))$ columns, and is a matrix over an alphabet of size at most $B+k$.

It is easy to see that our kernelization algorithm is polynomial and this concludes the proof.
\end{proof}

Theorem~\ref{thm:kernel} leads to the question whether \probFBClust admits a polynomial kernel when parameterized by $k$ and $B$ with the assumption that $\alpha$ is a fixed constant. A more general question is whether there are polynomial kernel for \probCClust, \probBClust and \probFBClust parameterized by $k$ and $B$. Notice that \probClust has a polynomial kernel for this parmeterization~\cite[Theorem~2]{FominGP20}. Another direction of research is to investigate kernels of other types. 
Are there polynomial  \emph{Turing kernels} and do these problem admit polynomial \emph{lossy kernels}, that is, 
 approximative kernels? (We refer to the book~\cite{FominLSZ19} for the definition of the notions.)

\end{document}